\documentclass[11pt]{article}

\usepackage{float}
\usepackage[dvips]{graphicx}
\usepackage{psfrag}

\usepackage{amsmath,amsthm,amsfonts,amssymb,amscd}
\usepackage[table]{xcolor}
\usepackage{mathrsfs}
\usepackage{upgreek}
\usepackage{stackrel}
\usepackage{tipa}
\usepackage{accents}
\usepackage[multiple]{footmisc}

\numberwithin{equation}{section}

\newcommand{\tdb}{\textup{\textcrd}}

\newcommand{\nfi}{\varphi}
\newcommand{\cp}{\partial}

\newcommand{\bs}{\boldsymbol}
\newcommand{\lp}{\left(}
\newcommand{\rp}{\right)}
\newcommand{\td}{\textup{d}}
\newcommand{\cu}{\mathcal{U}}
\newcommand{\cv}{\mathcal{V}}
\newcommand{\e}{\eta}
\newcommand{\be}{\bar{\eta}}

\newcommand{\Mp}{\mathcal{M^+}}
\newcommand{\Ml}{\mathcal{M^-}}

\newcommand{\lb}{\left\lbrace}
\newcommand{\rb}{\right\rbrace}
\newcommand{\lc}{\left[}
\newcommand{\rc}{\right]}
\newcommand{\ld}{\left.}
\newcommand{\rd}{\right.}
\newcommand{\lv}{\left\vert}
\newcommand{\rv}{\right\vert}
\newcommand{\la}{\langle}
\newcommand{\ra}{\rangle}
\newcommand{\rag}{\rangle_{g}}
\newcommand{\ragm}{\rangle_{g^-}}
\newcommand{\ragp}{\rangle_{g^+}}

\newcommand{\chr}{\Upupsilon}
\newcommand{\Mpm}{\mathcal{M}^{\pm}}
\newcommand{\tdo}{\tdb}
\newcommand{\lieo}{\mathsterling}

\newcommand{\uhat}{\underaccent{\check}}

\def\hup{h^{\sharp}}
\def\h{\widehat{h}}

\newcommand{\wt}{\widetilde}
\newcommand{\qhm}{\vert q\vert_{h^-}^2}

\makeatletter
\newsavebox\myboxA
\newsavebox\myboxB
\newlength\mylenA

\newcommand*\xoverline[2][0.75]{%
    \sbox{\myboxA}{$\m@th#2$}%
    \setbox\myboxB\null
    \ht\myboxB=\ht\myboxA%
    \dp\myboxB=\dp\myboxA%
    \wd\myboxB=#1\wd\myboxA
    \sbox\myboxB{$\m@th\overline{\copy\myboxB}$}
    \setlength\mylenA{\the\wd\myboxA}
    \addtolength\mylenA{-\the\wd\myboxB}%
    \ifdim\wd\myboxB<\wd\myboxA%
       \rlap{\hskip 0.5\mylenA\usebox\myboxB}{\usebox\myboxA}%
    \else
        \hskip -0.5\mylenA\rlap{\usebox\myboxA}{\hskip 0.5\mylenA\usebox\myboxB}%
    \fi}
\makeatother

\usepackage{geometry}
\newtheorem{prop}{Proposition}
\newtheorem{lem}{Lemma}
\newtheorem{cor}{Corollary}
\newtheorem{defn}{Definition}
\newtheorem{rem}{Remark}

\title{Null shells: general matching across null boundaries and connection with cut-and-paste formalism}

\author{Miguel Manzano and Marc Mars \\ \\
Instituto de F\'{\i}sica Fundamental y Matem\'aticas, IUFFyM\\
Universidad de Salamanca\\
Plaza de la Merced s/n \\
37008 Salamanca, Spain\\
}

 \newcounter{mnotecount}

 \newcommand{\mnote}[1]
 {\protect{\stepcounter{mnotecount}}$^{\mbox{\footnotesize
 $\,\bullet$\themnotecount}}$ \marginpar{
 \raggedright\tiny\em
 $\,\bullet$\themnotecount: #1} }

\begin{document}

\maketitle

\begin{abstract}
Null shells are a useful geometric construction to study the propagation of infinitesimally thin concentrations of massless particles or impulsive waves. In this paper, we determine and study the necessary and sufficient conditions for the matching of two spacetimes with respective null embedded hypersurfaces as boundaries. Whenever the matching is possible, it is shown to depend on a diffeomorphism between the set of null generators in each boundary and a scalar function, called step function, that determines a shift of points along the null generators. Generically there exists at most one possible matching but in some circumstances this is not so. When the null boundaries are totally geodesic, the point-to-point identification between them introduces a freedom whose nature and consequences are analyzed in detail. The expression for the energy-momentum tensor of a general null shell is also derived. Finally, we find the most general shell (with non-zero energy, energy flux and pressure) that can be generated by matching two Minkowski regions across a null hyperplane. This allows us to show how the original Penrose's cut-and-paste construction fits and connects with the standard matching formalism. 
\end{abstract}

\section{Introduction}

In General Relativity, thin shells (also called surface layers) are idealized objects introduced to describe sufficiently narrow concentrations of matter or energy such that they can be considered to be located on a hypersurface. Depending on the causal character of the hypersurface, thin shells can be null, timelike, spacelike or mixed. In this paper, we shall focus on 
null shells, which have been proven useful for modelling infinitesimally thin concentrations of massless particles or impulsive waves. Null shells possess their own gravity and hence affect the spacetime geometry. To study this, and also to understand how null shells can be constructed, 
  it is necessary to consider two spacetime regions (one at each side of the null shell) that must be suitably matched according to the corresponding theory.

  There exist two different approaches to generate null thin shells: the so-called ``cut-and-paste" method, introduced by Penrose \cite{dewitt1968battelle}, \cite{Penrose:1972xrn}, and the matching theory firstly developed by Darmois \cite{darmois1927memorial}. The ``cut-and-paste'' method describes the shell by means of a metric with a Dirac delta distribution with support on the shell. In these coordinates, the metric is therefore very singular, and standard tensor distributional calculus is not sufficient to study its geometry. However, by a suitable change of coordinates the metric becomes continuous and the method can be reinterpreted as follows.
  Given a spacetime $\lp\mathcal{M},g\rp$ containing an embedded null hypersurface $\Omega$, the cut-and-paste procedure uses lightlike coordinates adapted to $\Omega$. Then, $\Omega$ is removed by a cut, which leaves two separated manifolds $\lp\Mpm,g^{\pm}\rp$ corresponding to both sides of $\Omega$. Finally, those regions are reattached by identifying their boundaries so that there exists a jump on the coordinates when crossing the matching null hypersurface.
  This jump is responsible for the appearance of the Dirac delta term in the metric, which is interpreted as a concentration
  of matter and energy located on the matching hypersurface. By means of this useful geometrical approach, Penrose was able to study certain classes of impulsive plane-fronted and spherically-fronted waves in Minkowski's backgrounds. Later works of works of
  Podolsk{\`y}
  et al. \cite{podolsky1999nonexpanding}, \cite{podolsky2017penrose}, \cite{podolsky2019cut}
  apply this method to generate spacetimes whose metric again contains a Dirac delta function with support on the null hypersurface. The most general construction so far describes pp-waves with additional gyratonic terms \cite{podolsky2017penrose}.
  The cut-and-paste method is, by construction, strongly linked to the use of appropriate coordinate systems adapted both to the spacetime and to the null hypersurface where the cut is performed.

The second approach is the matching theory of Darmois, consisting of considering two spacetimes $\lp\Mpm,g^{\pm}\rp$ with respective differentiable boundaries $\Omega^{\pm}$ and generating a new spacetime by identifying the boundary points and the full tangent spaces defined on $\Omega^{\pm}$. The initial works by Darmois \cite{darmois1927memorial} 
were focused on the timelike and the spacelike cases. In any successful matching the spacetimes $\lp\Mpm,g^{\pm}\rp$ must satisfy the so-called preliminary junction conditions (see e.g. \cite{clarke1987junction}, \cite{mars1993geometry}), which force the boundaries $\Omega^{\pm}$ to be isometric with respect to their respective induced metrics. The resulting spacetime verifies the well-known Israel equations \cite{israel1966singular}, and describes a thin layer with material content given by the jump on the extrinsic curvatures. The null case is initially addressed by Barrab{\'e}s-Israel \cite{barrabes1991thin} and Barrab{\'e}s-Hogan \cite{barrabes2003singular} (see also
\cite{poisson2002reformulation} for a useful reformulation),
while the general causal character was developed 
in  \cite{mars1993geometry}, with an important aspect clarified in
\cite{mars2007lorentzian}
, \cite{poisson2004relativist}.
In 2013, the Israel equations have been firstly deduced for thin shells in a completely general (even variable) causal character \cite{mars2013constraint}. This has been achieved by means of a new formalism involving so-called
  \textit{hypersurface data}
  that allows one to abstractly analyse hypersurfaces of arbitrary signature in pseudo-riemannian manifolds \cite{mars2013constraint}, \cite{mars2020hypersurface}.
  The Israel equations  
   as well as additional equations
  involving curvature terms have been obtained recently in \cite{senovilla2018equations} also for shells of general causal character
    using the
    formalism of tensor distributional calculus. Many explicit examples of null shells in specific situations
    have been discussed in the literature, often by imposing additional symmetries such as spherical symmetry. We refer to
    \cite{carballo2021inner}, 
    %
    \cite{bhattacharjee2020memory},
    \cite{nikitin2018stability},
    \cite{chapman2018holographic},
    \cite{binetruy2018closed},
    \cite{kokubu2018energy},
    \cite{fairoos2017massless} and references therein for recent examples.

Despite the long history of both approaches, to the best of our knowledge there does not exist any systematic analysis of the connection between the cut-and-paste constructions and the matching theory of Darmois.
Our aim in this paper is two-fold. Firstly, we analyze in detail the Darmois matching theory
  across null hypersurfaces and obtain in full generality both the conditions that need to be satisfied for the matching and the energy-momentum contents of the resulting shell in terms
  the geometry of the ambient spaces and the identification of boundaries.
  We keep the choices of various geometric quantities on both sides very flexible, so that the general results can be applied easily to many different situations. This analysis is performed with our second goal in mind, namely to understand how the cut-and-paste construction fits into the matching framework.
Following the terminology of \cite{mars2013constraint}, \cite{mars2020hypersurface}, the preliminary junction conditions impose that the metric hypersurface data corresponding to the respective boundaries $\Omega^{\pm}$ must coincide. These metric data depend on both the geometry of the ambient spacetime and the way the hypersurface is embedded. In the cut-and-paste constructions, one deals with a single spacetime and a single null hypersurface. After the cut along this hypersurface
  a reorganization of points must be performed on one side. For this construction to fit into the general matching framework, it must be the case 
  the new embedding  defines the same metric hypersurface data, and this automatically restricts the set of possible reorganizations of points. The compatibility between the cut-and-paste method and the matching theory by Darmois will be achieved whenever the redistribution of points performed in the first keeps the metric hypersurface data invariant. We show in this paper that the original
cut-and-paste construction by Penrose \cite{penrose1965remarkable}, \cite{penrose1968twistor}, \cite{dewitt1968battelle}, \cite{Penrose:1972xrn}, namely the plane-fronted impulsive wave, satisfies this compatibility in full, and hence
  can be understood also in terms of a standard matching procedure.  Having fitted this cut-and-paste construction into a more general framework allows us to
  generalize the results of Penrose by obtaining the most general shell (with non-vanishing energy, energy flux and pressure) generated by matching two Minkowski regions. 
This generalization allows us to understand how the shell pressure affects the organization of the boundary points.
It turns out
that a positive (resp. negative) pressure produces compression (resp. stretching) of points on one of the sides and that this entails energy increase (resp. decrease). The  energy of the shell shows an accumulative behaviour and
only varies when this compression/stretching is taking place.
In a subsequent work we intend to
analyze 
in terms of the Darmois matching conditions (and eventually extend to more general matter contents)
other explicit cut-and-paste constructions such as the spherical fronted waves of Penrose or the constructions by Podolsk{\`y} and collaborators
\cite{podolsky1999nonexpanding}, \cite{podolsky2014gyratonic}, \cite{podolsky2017penrose}, \cite{podolsky2019cut}.



Given two spacetimes $\lp\Mpm,g^{\pm}\rp$ with null  boundaries $\Omega^{\pm}$, we demonstrate that the core problem of the existence of the matching lies on the solvability of one of the junction conditions, which establishes an isometry condition between any spatial section on $\Omega^-$ and the submanifold of $\Omega^+$ with which it is identified. This isometry must be universal along each null generator. More precisely, the matching is possible if and only if, in addition to a necessary causal restriction at the boundaries $\Omega^{\pm}$, there exists a diffeomorphism $\Psi$ between the set of null generators of $\Omega^-$ and the set of null generators of $\Omega^+$ as well as a scalar function $H$, called step function, which geometrically corresponds to a shift along the null generators. The map $\Psi$ and the step function $H$ fix completely the identification $\Phi$ between the boundaries. The fundamental restriction on these maps arising from the matching conditions is that any spacelike section on $\Omega^-$ must be isometric to its image under $\Phi$. We prove that the whole matching is determined by the step function and the diffeomorphism $\Psi$ and obtain the explicit form of the energy-momentum tensor of a general null shell in terms of them and ambient geometrical objects defined on $\Omega^{\pm}$.

It is to be expected that, generically, two given spacetimes $\lp\Mpm,g^{\pm}\rp$ with null boundaries cannot be matched across them. Our results are in agreement with this expectation. We materialize this rigidity in the matching in a uniqueness statement, namely that when the matching between $\lp\Mpm,g^{\pm}\rp$ is feasible, there exists generically one unique way of matching, i.e. there is only one suitable identification of the boundary points and the full tangent spaces so that the metric hypersurface data of both sides agree. However, in some cases one can perform not only more than one matching but even infinite. This freedom is discussed in section \ref{secstepfunct}, happens when the boundaries $\Omega^{\pm}$ are totally geodesic embedded null hypersurfaces and translates into the existence of an infinite set of step functions $H$ leading to successful matchings. Different step functions give rise to different shells, with different energy-momentum tensors. Imposing specific conditions of the energy and momentum on the shell equations provide differential equations that restrict the form of the corresponding step functions.

The organization of the paper is as follows. Section 2 is devoted to
recalling and extending several facts on the geometry of null hypersurfaces. We introduce some geometric objects that are required to develop the matching formalism, together with several identities that they satisfy. As already mentioned, we let these geometric quantities to be  flexible so that the results can accommodate  different needs when an explicit matching is to be performed.  In Section 3, we derive the shell junction conditions and determine
    the constraint that arises from the condition that the two rigging vector fields that define the identification of transversal directions point into the same side of the matching hypersurface.
  Section 4 is devoted to the derivation and properties of the step function and the diffeomorphism $\Psi$ mentioned before. We explicitly compute the components of the energy-momentum tensor of a general null shell in Section 5. 
  The general behaviour of the energy-momentum tensor under transformations of the riggings is also discussed, which allows us to perform a non-trivial check on our result (see Appendix A). In Section 6, we find the most general null shell in the context of the matching of two Minkowski regions 
and recover the results from the first cut-and-paste construction by Penrose (see e.g. \cite{dewitt1968battelle}) as a particular case.


Throughout this paper we assume the spacetimes to be $\lp n+1\rp-$dimensional and adopt the following index convention:
\begin{align}
&\alpha,\beta,...=0,1,...,n, && i,j,...=1,2,...,n, && I,J,...=2,3,...,n.
\end{align}

\section{Geometric structure of null hypersurfaces}
As already indicated in the introduction, this paper is devoted to studying the matching between two spacetimes with null boundaries. In order to address the problem, we need to recall some notions on the geometry of null hypersurfaces. General references for the topic are
\cite{galloway2004null}, \cite{gourgoulhon20063+}.
\begin{defn}
  
  Let $\lp\mathcal{M},g\rp$ be an $\lp n+1\rp-$dimensional pseudo-riemannian manifold and $\Sigma$ a manifold of dimension $n$. An embedded null hypersurface is a subset $\Omega\subset\mathcal{M}$ satisfying that there exists an embedding $\bs{\Phi}:\Sigma\rightarrow\mathcal{M}$ such that $\bs{\Phi}\lp \Sigma\rp=\Omega$ and that the first fundamental form $\gamma$ of $\Sigma$, defined by $\gamma = \Phi^*\lp g\rp$, is degenerate.
\end{defn}
As usual, we use $\Phi^*$ and $\td\Phi$ to refer to the pull-back and push-forward of $\Phi$ respectively, $T_p\mathcal{M}$ denotes the tangent space of $\mathcal{M}$ at the point $p \in\mathcal{M} $ and $T^*_p\mathcal{M}$ its dual, or cotangent, space, and similarly for $\Sigma$ and $\Omega$. Given two points $q\in\Sigma$ and $p=\Phi\lp q\rp\in\Omega$, the tangent plane to the embedded hypersurface, $T_{p}\Omega$, is the $n-$dimensional space defined as $T_{p}\Omega=\td\Phi\vert_q\lp T_q\Sigma\rp$. Its orthogonal space, i.e. the space of vectors that are orthogonal to all those on $T_p\Omega$, is written as $\lp T_p\Omega\rp^{\perp}$. As always, $T\mathcal{M}$, $T\Omega$ and $T\Sigma$ denote the corresponding tangent bundles.

It is well-known (see e.g. \cite{gourgoulhon20063+})
that there exists only one degenerate direction on $\Omega$ and that all other directions tangent to the null hypersurface are spacelike. Let $\hat{k}$ be any nowhere zero vector field along the degenerate direction of $\Omega$, which by definition means
\begin{equation}
\label{eqA1}
\hat{k}\vert_p \neq 0, \qquad \la \hat{k},X\rag\vert_p=0, 
\end{equation}
for any vector $\ld X\rv_p\in T_p\Omega$. Thus, $\hat{k}\vert_p\in T_p\Omega\cap\lp T_p\Omega\rp^{\perp}$. Since $\td\Phi$ is of maximal rank, the dimension of $T_{p}\Omega$ is $n$ so the dimension of $\lp T_p\Omega\rp^{\perp}$ is 1. It follows that $\lp T_p\Omega\rp^{\perp}\subset T_p\Omega$ and hence $\la \hat{k}\vert_p\rangle=\lp T_p\Omega\rp^{\perp}$ and all the vectors in $\lp T_p\Omega\rp^{\perp}$ are null. We let $T \Omega^{\perp} =\bigcup_p (T_p \Omega)^{\perp}$. It is clear that this is a subbundle to $T \Omega$. A null generator $k$ of $\Omega$ is defined to be a nowhere zero
section\footnote{$\Gamma\lp E\rp$ denotes the sections of any bundle $E$.}   $ k\in \Gamma \lp T\Omega^{\perp}\rp$. Null generators are all proportional to each other.

It is well-known (see e.g. equation (2.21) in \cite{gourgoulhon20063+}) that a null generator $k$
is necessarily geodesic (not necessarily affinely parametrized). Its surface gravity $\kappa_{k}$  is
defined by
\begin{equation}
\label{kappadef}
\nabla_kk=\kappa_kk,
\end{equation}
where $\nabla$ denotes the Levi-Civita covariant derivative of $g$.

The matching requires an identification between points and tangent spaces of two embedded null hypersurfaces. It will be convenient to describe these identifications in terms of sections of the null hypersurface. In fact, they will also help clarifying the physical properties of the matching. We define the concepts of spacelike section, tangent plane and foliation of an embedded null hypersurface as follows.
\begin{defn}
\label{def1}
Let $\Omega$ be an embedded null hypersurface and $k$ a null generator. Suppose the existence of a function $s:\Omega\rightarrow\mathbb{R}$ such that $ k\lp s\rp\neq0$ everywhere in $\Omega$. Then, the section $S_{s_0}$ is defined as the subset
\begin{equation}
S_{s_0}:=\lb p\in\Omega\textup{ }\vert\textup{ }s\lp p\rp=s_0,\textup{ }s_0\in\mathbb{R}\rb.
\end{equation}
Given $p \in \Omega$ and the section $S_{s\lp p\rp} \subset \Omega$ that contains $p$, the tangent plane $T_pS_{s\lp p\rp}$ is defined as
\begin{equation}
T_pS_{s\lp p\rp}:=\lb X\in T_p\Omega\textup{ }\vert\textup{ } X\lp s\rp=0\rb.
\end{equation}
All the above requirements guaranteed, the family of spacelike sections $\lb S_s\rb$ define a foliation of $\Omega$ given by the levels of $s$, i.e. the subsets of constant $s$.
\end{defn}
We emphasize that the existence of the function $s$ in Definition \ref{def1}
entails a global restriction on $\Omega$, since it automatically follows that all sections associated to $s$ are diffeomorphic to each other, and that the
topology of $\Omega$ is $S_{s_0} \times \mathbb{R}$.  We shall assume this global restriction throughout the paper. We
stress, however, that locally near any point, the existence of the function $s$ can always be granted, so all the results in the paper of a purely local nature are valid in full generality.

Given one foliation, there exists one unique null generator $k$ satisfying that $k\lp s\rp=1$. This choice will be always assumed in this paper.

A fundamental result of submanifolds is that the Lie bracket of two tangent vectors is also tangent (see e.g. \cite{lee2003introduction}). Thus $X,Y\in \Gamma\lp T\Omega\rp$  implies
$\lc X,Z\rc \in \Gamma\lp T\Omega\rp$. Given some neighbourhood $\cu\subset\Omega$, when $X,Y$ are in addition tangent to the sections defined by $s$
which we write $X,Y\in \stackbin[p\in\mathcal{U}]{}{\bigcup}T_pS_{s\lp p\rp}$, then also
$\lc X,Y\rc\in \stackbin[p\in\mathcal{U}]{}{\bigcup}T_pS_{s\lp p\rp}$. From now on, given a foliation in $\Omega$ defined by $s$, we will use $\lb k, v_I\rb$ to refer to any basis of $T\Omega$ satisfying the following properties:
\begin{equation}
\label{basis}
\begin{array}{cl}
\textup{(A)} & k\textup{ is the null generator satisfying }k\lp s\rp=1.\\
\textup{(B)} & \textup{Each }v_I\textup{ is a vector field (necessarily spacelike) verifying that }v_I\vert_p\in T_pS_{s\lp p\rp}\textup{ at each }p\in\Omega.\\
\textup{(C)} & k\textup{ and }v_I\textup{ verify that }\lc k,v_I\rc=0\textup{ and }\lc v_I,v_J\rc=0.\\
\end{array}
\end{equation}
When $\Omega$ is the boundary of a  spacetime, it is automatically two-sided and hence it always admits
(see Lemma 1 in \cite{mars2013constraint}) a transversal vector field $L_0$, i.e. satisfying
$L_0 \notin T_p\Omega, \forall p \in \Omega$. If in addition $\Omega$ is null, this transverse vector field
can always be taken to be null everywhere. Indeed, $L_0$ being transversal means that $\la L_0, k \ra \neq 0$
everywhere and then
\begin{align*}
  L :=  L_0 - \frac{\la L_0, L_0\rag}{2 \la L_0, k \rag } k
\end{align*}
is both null and transversal.  We select one such vector
(null transversal vector fields are highly non-unique) and choose a basis $\lb k,v_I\rb$ of $T\Omega$. 
We introduce the $n$ scalar functions $\nfi$ and $\psi_I$  on $\Omega$ defined by 
\begin{equation}
\label{eqA30}
\nfi\lp p\rp:=-\ld\la L,k\rag\rv_p \qquad \psi_I\lp p\rp:=- \ld\la L,v_I\rag\rv_p
\end{equation}
and we observe that necessarily $\nfi\lp p\rp\neq0$. These functions obviously depend on the choice of  the basis $\lb k, v_I\rb$ and $L$. For the sake of simplicity we do not reflect this dependence in the notation. On the other hand, it may seem strange not to restrict $L$ to satisfy $\psi_I=0$, i.e. to be orthogonal to the leaves of the foliation. The reason is that there can be many cases where the most convenient choice of $L$ (e.g. to simplify the computations) does not verify $\psi_I=0$. Thus, we keep these functions completely free a priori. An explicit example  where
choosing $L$ not orthogonal to the leaves turned out to be useful appears in \cite{podolsky2017penrose}. Actually, the functions $\psi_I$
in that paper happen to be the currents $J\lp\cu,\e,\be\rp$ and $\bar{J}\lp\cu,\e,\be\rp$ which play a fundamental role in the
physical description of the impulsive gravitational wave associated to  the matching.

\subsection{Tensors on an embedded null hypersurface $\Omega$}

Since the first fundamental form is degenerate (i.e. it does not admit an inverse), there is no natural way of raising and lowering indices of tensors on $\Omega$. The standard way of dealing with this difficulty is to introduce a quotient structure (see e.g. \cite{galloway2004null}) by defining for any $Z,W\in T_p\Omega$ the equivalence relation $\sim$ as 
\begin{equation}
Z\sim W\quad\Longleftrightarrow\quad Z-W=\beta k,
\end{equation}  
where $\beta\in\mathbb{R}$.
\begin{defn}
Let $\Omega$ be an embedded null hypersurface and $p\in\Omega$. Then the quotient vector space $T_p\Omega/k$ is defined as
\begin{equation}
T_p\Omega/k:=\lb \bar{Z}:Z\in T_p\Omega\rb,
\end{equation}
where $\bar{Z}:=\lb X\in T_p\Omega:X\sim Z\rb$. The fiber bundle $T\Omega/k$ is the natural $\lp n-1\rp-$dimensional vector space given by
\begin{equation}
T\Omega/k:=\stackbin[p\in\Omega]{}{\bigcup}T_p\Omega/k.
\end{equation}
\end{defn}
This quotient structure of $T_p\Omega$ allows to construct a metric and a second fundamental form on $T\Omega/k$. The metric, denoted by $\h$, is the symmetric $2-$covariant tensor defined by 
\begin{equation}
\h\lp \bar{Z},\bar{W}\rp\vert_p:=\ld\la Z,W\rag\rv_p ,
\end{equation}
where $\bar{Z},\bar{W}\in T_p\Omega/k$. The tensor is well-defined because the
right-hand side is independent of the representatives $Z \in \bar{Z},
W \in \bar{W}$. Besides, for any non-zero
$\bar{Y}\in T_p\Omega/k$ (i.e. classes associated to
spacelike directions $Y$), it is satisfied that $ \h\lp \bar{Y},\bar{Y}\rp\vert_p=\la Y,Y\rag\vert_p>0$. Thus, $\h$ is a positive definite metric. Given $p_0\in\Omega$, a section $S_{s\lp p_0\rp}$ and any $p \in
S_{s\lp p_0\rp}$, $\h|_{p}$ is isometric to the induced metric $h$ of $S_{s\lp p_0\rp}$ at $p$. Indeed, the map $T_p S_{s\lp p_0\rp}\longrightarrow T_{p_0}\Omega/k$ defined by $X\longrightarrow\bar{X}$ is an isomorphism\footnote{Both spaces have the same dimension and the kernel is obviously $\{0\}$ because $k\lp s\rp\neq0$.} and for any two vectors $Z,W \in T_p S_{s\lp p_0\rp}$ it holds 
\begin{align}
\h\lp \bar{Z},\bar{W}\rp\vert_p=\la Z+a k,W+bk\rag\vert_p=\la Z,W\rag\vert_p\equiv h\lp Z,W\rp\vert_p.
\label{equalityhbarh}
\end{align}
Thus, $h$ is also positive definite and we denote by $\hup$ its associated contravariant metric. Their components in a basis $\lb v_I\vert_p\rb$ of $T_p S_{s\lp p_0\rp}$ and its corresponding dual $\lb \omega^I\vert_p\rb$ are denoted by $h_{IJ}$ and $h^{IJ}$ respectively. We use these tensors to lower  and raise  capital
indices, irrespectively of whether they are tensorial,
e.g. in $\psi_I$, or identify elements in a set, e.g. in $v_I$.

The second fundamental form with respect to $k$, denoted by $\bs{\hat{\chi}}^k$, is the $2-$covariant tensor on $T_p\Omega/k$ defined by 
\begin{equation}
\begin{array}{lcl}
\bs{\hat{\chi}}^k\lp \bar{Z},\bar{W}\rp\vert_p:=\la \nabla_{Z}k,W\rag \vert_p, & &\bar{Z},\bar{W}\in T_p\Omega/k.
\end{array}
\end{equation}
Again, this tensor is well-defined. i.e. independent of the choice of representatives, and  closely related to the second fundamental form $\bs{\chi}^k$ of $S_{s\lp p_0\rp}$ with respect to the normal $k$. Specifically, for any two vectors $Z,W\in T_pS_{s\lp p\rp}$, it holds
\begin{eqnarray}
\nonumber \bs{\hat{\chi}}^k\lp \bar{Z},\bar{W}\rp\vert_p&=&\la \nabla_{Z+ak}k,W+bk\rag\vert_p=\la \nabla_{Z+ak}k,W\rag\vert_p=\la \nabla_{Z}k,W\rag\vert_p+ a\la \nabla_{k}k,W\rag\vert_p\\
\label{moreeqs}&=&\la \nabla_{Z}k,W\rag\vert_p+a\kappa_k\la k,W\rag\vert_p=\la \nabla_{Z}k,W\rag\vert_p\equiv\bs{\chi}^k\lp Z,W\rp\vert_p.
\end{eqnarray}
For later use, we recall the following relation \cite{gourgoulhon20063+} between the rate of change of the induced metric along $k$ and the second fundamental form of the section 
\begin{equation}
\label{eqA28}
k\lp  h\lp v_K,v_I\rp\rp=2\bs{\chi}^k\lp v_I,v_K\rp.
\end{equation}
This identity relies on the basis vector fields $v_I$ verifying $\lc k,v_I\rc=0$ and uses the fact that $\bs{\chi}^k$ is symmetric. 

With respect to the transversal null direction $L$, we define the 2-covariant tensor $\bs{\Theta}^L$ and the one-form $\bs{\sigma}_L$ at $p \in S_{s\lp p_0\rp}$ by
\begin{equation}
\label{somedefs}
\begin{array}{lcl}
\bs{\Theta}^L\lp Z, W\rp\vert_p:=\la \nabla_ZL,W\rag\vert_p, & &\bs{\sigma}_L\lp Z\rp\vert_p:=\dfrac{1}{\nfi}\la \nabla_Zk,L\rag\vert_p,
\end{array}
\end{equation}
where $Z,W\in T_pS_{s\lp p_0\rp}$ and $\nfi$ is given by (\ref{eqA30}). Note that since we are not assuming $L$ to be normal to the section, $\bs{\Theta}^L$ is {\emph{not}} one of the second fundamental forms of the section. In fact, this tensor is not symmetric in general.

\subsection{Covariant derivatives along directions tangent to an embedded null hypersurface $\Omega$}

For later purposes, it is convenient to provide the explicit form of some covariant derivatives with respect to the vector fields $k$ and $v_I$. Since $\nabla_{k}k$ is given by (\ref{kappadef}) and $\nabla_{k}v_I=\nabla_{v_I}k$ (c.f. (\ref{basis})), we only require $\nabla_{v_I}v_J$, $\nabla_{v_I}k$, $\nabla_{v_I}L$ and $\nabla_kL$.
When $L$ is normal to the sections, the corresponding expressions can be found e.g.\ in \cite{roesch2016proof}, but we are not aware of a reference where the explicit expressions for general $L$ appear. Actually these expressions can be regarded as an
expanded form of equations (19) and (21) in
\cite{mars1993geometry}.
\begin{lem}
\label{lem1}
Let $\Omega$ be an embedded null hypersurface, $\lb S_s\rb$ a foliation of $\Omega$ defined by $s$ and $L$ a null vector field everywhere transversal to $\Omega$. Given a basis $\lb k,v_I\rb$ satisfying conditions (\ref{basis}), let $\lb \omega^I\vert_p\rb$ be the basis of $T^*_pS_{s\lp p\rp}$ dual to $\lb v_I\vert_p\rb$. Then the derivatives $\nabla_{v_I}v_J$, $\nabla_{v_I}k$, $\nabla_{v_I}L$ and $\nabla_{k}L$ take the following form:
\begin{align}
\label{eqA27a}&\nabla_{v_I}v_J=\frac{1}{\nfi}\bs{\chi}^k\lp v_I,v_J\rp L+\frac{1}{\nfi}\lp v_I\lp\psi_J\rp+\bs{\Theta}^L\lp v_I,v_J\rp-\chr_{JI}^K\psi_K\rp k+\chr_{JI}^Kv_K,\\
\label{eqA27b}&\nabla_{v_I}k=-\lp \bs{\sigma}_L\lp v_I\rp+\frac{1}{\nfi}\psi^B\bs{\chi}^k\lp v_I,v_B\rp\rp k+\bs{\chi}^k\lp v_I,v^B\rp v_B,\\
\label{eqA27c}&\nabla_{v_I}L=\eta_I  L-\frac{1}{\nfi}\psi^J\lp \eta_I \psi_J+\bs{\Theta}^{L}\lp v_I,v_J\rp\rp k+\lp \eta_I \psi^J+\bs{\Theta}^{L}\lp v_I,v^J\rp\rp v_J,\\
\label{derkL}&\nabla_kL=\lp \dfrac{k\lp\nfi\rp}{\nfi}-\kappa_k\rp \lp L-\dfrac{\psi^I\psi_I}{\nfi}k+\psi^I v_I\rp +\lp k\lp\psi_I\rp+\nfi\bs{\sigma}_L\lp v_I\rp\rp\lp \dfrac{\psi^I}{\nfi}k-v^I\rp,
\end{align}
where $\chr_{JI}^K$ is given by 
\begin{equation}
\label{eqA48}
\chr_{JI}^K=\lp\la v^K,\nabla_{v_I}v_J\rag+\frac{1}{\nfi}\psi^K\bs{\chi}^k\lp v_I,v_J\rp\rp,
\end{equation}
and $\eta_I$ is defined as
\begin{equation}
\label{eqA38}
\eta_I:=\lp \frac{1}{\nfi}v_I\lp\nfi\rp+\bs{\sigma}_L\lp v_I\rp\rp.
\end{equation}
\end{lem}
\begin{rem}Note that the vector field $ L-\dfrac{\psi^I\psi_I}{\nfi}k+\psi^I v_I$ is orthogonal to both $v_J$ and $L$, whereas $\dfrac{\psi_I}{\nfi}k-v_I$ is orthogonal to $L$.
\end{rem}

\begin{proof}
We start with $\nabla_{v_I}v_J$. For suitable scalar functions $\alpha_{IJ}$, $\beta_{IJ}$ and $\chr_{JI}^K$, this derivative can be expressed as $\nabla_{v_I}v_J=\alpha_{IJ}L+\beta_{IJ}k+\chr_{JI}^Kv_K$. Using (\ref{eqA30}), it follows that
\begin{align}
\nonumber &\la k,\nabla_{v_I}v_J\rag=-\alpha_{IJ}\nfi,&&\Longleftrightarrow&&\alpha_{IJ}=-\frac{1}{\nfi}\la k,\nabla_{v_I}v_J\rag,&\\
&\la v_L,\nabla_{v_I}v_J\rag=-\alpha_{IJ}\psi_L+\chr_{JI}^Kh_{KL},&&\Longleftrightarrow &&\chr_{JI}^K=\la v^K,\nabla_{v_I}v_J\rag+\alpha_{IJ}\psi^K,&\\
\nonumber &\la L,\nabla_{v_I}v_J\rag=-\beta_{IJ}\nfi-\chr_{JI}^K\psi_K, &&\Longleftrightarrow && \beta_{IJ}=-\frac{1}{\nfi}\lp\la L,\nabla_{v_I}v_J\rag+\chr_{JI}^K\psi_K\rp,&
\end{align}
so that, taking into account (\ref{moreeqs}) and (\ref{somedefs})
\begin{align}
\label{eqA33}
\nonumber &\alpha_{IJ}=\frac{1}{\nfi}\la v_J,\nabla_{v_I}k\rag= \frac{1}{\nfi}\bs{\chi}^k\lp v_I,v_J\rp,\\
&\chr_{JI}^K=\la v^K,\nabla_{v_I}v_J\rag+\frac{1}{\nfi}\psi^K\bs{\chi}^k\lp v_I,v_J\rp,\\
\nonumber &\beta_{IJ}=-\frac{1}{\nfi}\lp v_I\lp\la L,v_J\rag\rp-\la \nabla_{v_I}L,v_J\rag+\chr_{JI}^K\psi_K\rp=\frac{1}{\nfi}\lp v_I\lp\psi_J\rp+\bs{\Theta}^L\lp v_I,v_J\rp-\chr_{JI}^K\psi_K\rp.
\end{align}
Concerning $\nabla_{v_I}k$, we decompose $\nabla_{v_I}k=\alpha_IL+\beta_Ik+{\varepsilon}_I^Lv_L$ and find 
\begin{align}
\nonumber &{\alpha}_I=-\frac{1}{\nfi}\la k,\nabla_{v_I}k\rag=0,\\
&{\varepsilon}_I^L=\la v^L,\nabla_{v_I}k\rag=\bs{\chi}^k\lp v_I,v^L\rp,\\
\nonumber &{\beta}_I=-\frac{1}{\nfi}\lp\la L,\nabla_{v_I}k\rag+{\varepsilon}_I^K\psi_K\rp=-\lp\bs{\sigma}_L\lp v_I\rp+\frac{1}{\nfi}\psi^J \bs{\chi}^k\lp v_I,v_J\rp\rp.
\end{align}
Substituting in $\nabla_{v_I}k$ gives (\ref{eqA27b}). On the other hand, decomposing $\nabla_{v_I}L=\mu_IL+\nu_Ik+{\rho}_I^Lv_L$ one obtains
\begin{align}
\nonumber &\mu_I=-\frac{1}{\nfi}\la k,\nabla_{v_I}L\rag=-\dfrac{1}{\nfi}v_I\lp \la k,L\rangle_g\rp+\dfrac{1}{\nfi}\la \nabla_{v_I}k,L\rangle_g=\dfrac{1}{\nfi}v_I\lp \nfi\rp+\dfrac{1}{\nfi}\la \nabla_{v_I}k,L\rangle_g,\\
&{\rho}_I^L=\mu_I\psi^L+\la v^L,\nabla_{v_I}L\rag,\\
\nonumber &\nu_I=-\frac{1}{\nfi}\la L,\nabla_{v_I}L\rag-\frac{1}{\nfi}{\rho}_I^K\psi_K=-\frac{1}{\nfi}{\rho}_I^K\psi_K.
\end{align}
Using the definitions (\ref{somedefs}) and (\ref{eqA38}) and inserting the results into $\nabla_{v_I}L$ proves (\ref{eqA27c}). Finally, writing $\nabla_kL=aL+bk+c^Iv_I$ yields 
\begin{align}
\label{d1}\la k,\nabla_kL\rangle_g&=-k\lp\nfi\rp+\kappa_k\nfi=-a\nfi,\\
\label{d2}\la L,\nabla_kL\rangle_g&=0=-b\nfi-c^I\psi_I,\\
\label{d3}\la v_J,\nabla_kL\rangle_g&=-k\lp\psi_J\rp-\la L,\nabla_kv_J\rangle_g=-k\lp\psi_J\rp-\nfi\bs{\sigma}_L\lp v_J\rp=-a\psi_J+c_J.
\end{align}
Equation (\ref{d1}) immediately provides $a$, while from (\ref{d2}) one gets $b=-\frac{c^I\psi_I}{\nfi}$. Multiplying (\ref{d3}) by $h^{JK}$ gives $c^I=\lp\frac{k\lp\nfi\rp}{\nfi}-\kappa_k\rp\psi^I-h^{IJ} k\lp\psi_J\rp-\nfi\bs{\sigma}_L\lp v^I\rp$, and the substitution of $a,b,c^I$ on $\nabla_kL$ proves (\ref{derkL}).
\end{proof}

\section{General matching of spacetimes across a null hypersurface}
\subsection{Metric hypersurface data and hypersurface data}%

Let us start by introducing the concepts of metric hypersurface data and hypersurface data as defined in \cite{mars2013constraint}, \cite{mars2020hypersurface} and which provide a convenient setup to study shells of arbitrary causal character. Let $\Sigma$ be an $n-$dimensional manifold endowed with a $2-$symmetric covariant tensor $\gamma$, a $1-$form $\bs{\ell}$ and a scalar function $\ell^{(2)}$. The four-tuple $\lb\Sigma,\gamma,\bs{\ell},\ell^{(2)}\rb$ defines \textbf{metric hypersurface data} provided that the symmetric $2-$covariant tensor $\bs{\mathcal{A}}\vert_p$ on $T_p\Sigma\times\mathbb{R}$, called ambient metric and defined as
\begin{equation}
\label{ambientmetric}
\ld{\mathcal{A}}\rv_p\lp\lp W,a\rp,\lp Z,b\rp\rp:=\ld\gamma\rv_p\lp W,Z\rp+a\ld\bs{\ell}\rv_p\lp Z\rp+b\ld\bs{\ell}\rv_p\lp W\rp+ab \ell^{(2)}\vert_p, \qquad W,Z\in T_p\Sigma, \qquad a,b\in\mathbb{R},
\end{equation}
has Lorentzian signature at every $p\in\Sigma$. The five-tuple $\lb \Sigma,\gamma,\bs{\ell},\ell^{(2)},\bs{Y}\rb$ defines \textbf{hypersurface data} if $\bs{Y}$ is a symmetric $2-$covariant tensor field along $\Sigma$ and $\lb \Sigma,\gamma,\bs{\ell},\ell^{(2)}\rb$ is metric hypersurface data.

The abstract notion of (metric) hypersurface data connects to the geometry of hypersurfaces via the concept of ``embedded (metric) hypersurface data" defined as follows \cite{mars2020hypersurface}. The data $\lb\Sigma,\gamma,\bs{\ell},\ell^{(2)}\rb$ is \textbf{embedded} in a pseudo-riemannian manifold $\lp \mathcal{M},g\rp$ of dimension $n+1$ if there exists an embedding $\Phi:\Sigma\rightarrow\mathcal{M}$ and a vector field $\xi$ along $\Phi\lp\Sigma\rp$ and everywhere transversal to $\Sigma$, known as rigging, satisfying 
\begin{equation}
\label{emhd}
\Phi^*\lp g\rp=\gamma, \qquad\Phi^*\lp g\lp\xi,\cdot\rp\rp=\bs{\ell}, \qquad\Phi^*\lp g\lp\xi,\xi\rp\rp=\ell^{(2)}.
\end{equation}
The hypersurface data $\lb\Sigma,\gamma,\bs{\ell},\ell^{(2)},\bs{Y}\rb$ is embedded if besides $\lb\Sigma,\gamma,\bs{\ell},\ell^{(2)}\rb$ being embedded it holds
\begin{equation}
\label{eqcf}
\dfrac{1}{2}\Phi^*\lp \mathscr{L}_{\xi}g\rp=\bs{Y}.
\end{equation}
Of course, given some embedded hypersurface data with degenerate $\gamma$, the subset $\Phi\lp\Sigma\rp$ is an embedded null hypersurface on $\mathcal{M}$.

\subsection{Shell junction conditions}
Let $\lp \mathcal{M}^{\pm},g^{\pm}\rp$ be two $\lp n+1\rp-$dimensional time-oriented Lorentzian manifolds with respective
boundaries  $\Omega^{\pm}\subset\mathcal{M}^{\pm}$ and assume that $\Omega^{\pm}$ are (necessarily embedded) null
hypersurfaces. We make the global assumption on $\Omega^{\pm}$ mentioned above and require that
$\Omega^{\pm}$ admit a foliation $\big\{ S^{\pm}_{s^{\pm}}\}$ defined by a function $s^{\pm}$. We introduce a basis
$\{ L^{\pm},k^{\pm},v^{\pm}_I\big\}$ of $\Gamma\lp T\mathcal{M}^{\pm}\rp\vert_{\Omega^{\pm}}$ such that $L^{\pm}$ is a null vector field everywhere transversal to $\Omega^{\pm}$ and $\lb k^{\pm},v^{\pm}_I\rb$ satisfies conditions (\ref{basis}). For definiteness, we restrict $L^{\pm}$ and $k^{\pm}$ to be future directed. This entails no loss of generality and simplifies the discussion on existence of the matching, see below.

As discussed in the introduction, one of the  aims of this paper is to study what sort of spacetimes $(\mathcal{M},g)$ containing a null shell can be constructed by pasting $\mathcal{M}^{\pm}$ along $\Omega^{\pm}$.
It has also been mentioned that spacetimes of this kind have been studied either by means of the matching formalism or with the cut and paste method.
  In the standard matching formalism
the construction of a spacetime $(\mathcal{M},g)$ containing a null shell is only possible when the so-called shell junction conditions are satisfied. These are a set of equalities that provide information about the identification between points of $\Omega^+$ and $\Omega^-$ and between the tangent spaces $T\mathcal{M}^{\pm}$. $\Omega^{\pm}$ being embedded hypersurfaces, there exist two embedded metric hypersurface data $\big\{\Sigma^{\pm}, \gamma^{\pm},\bs{\ell}^{\pm},\ell_{\pm}^{(2)}\big\}$, two embeddings $\Phi^{\pm}$ satisfying that $\Phi^{\pm}\lp\Sigma^{\pm}\rp=\Omega^{\pm}$ and two vector fields $\xi^{\pm}\in \Gamma\lp T\mathcal{M}^{\pm}\rp$ everywhere transversal to $\Omega^{\pm}$ and verifying (\ref{emhd}). The key point is that the shell junction conditions impose that the two metric hypersurface data must coincide. Thus, from now on we will only deal with a single metric hypersurface data denoted by $\lb \Sigma,\gamma, \bs{\ell},\ell^{(2)}\rb$ both for $\Omega^+$ and $\Omega^-$.
Since the boundaries of $\Mpm$ are null, the tensor $\gamma$ is degenerate at every point with a single degeneration direction.

An important aspect to bear in mind is that, given two manifolds with boundary, determining whether they can be matched amounts to finding two embeddings of an abstract manifold $\Sigma$ onto their respective boundaries, in such a way that the shell matching conditions are fulfilled (i.e. that the corresponding metric hypersurface data agree). See Figure \ref{fig1} for a schematic picture of the construction.  The embeddings and the rigging vectors are not known or given a priori. In many circumstances such embeddings do not exist, and then the two spacetimes simply cannot be matched across their boundaries. In other cases, there exists multiple (even infinite) possible embeddings, giving rise to multiple joined spacetimes, which in general are different from each other (see the discussion of this point in section 4). When the shell junction conditions are satisfied, the geometry of the shell is determined by the jump of the transverse tensors $\bs{Y}^{\pm}$ defined as (c.f. (\ref{eqcf}))
\begin{equation}
\label{eqA21}
\bs{Y}^{\pm}:=\dfrac{1}{2}{\Phi^{\pm}}^*\lp \mathscr{L}_{\xi^{\pm}}g^{\pm}\rp.
\end{equation}
We therefore introduce the tensor
\begin{equation}
\label{eqA22}
\bs{V}:=\lc \bs{Y}\rc:=\bs{Y}^+-\bs{Y}^-.
\end{equation} 

\begin{figure}[h!]
\centering
\psfrag{A}{$\Mp$}
\psfrag{B}{$\Ml$}
\psfrag{C}{$\Omega^-$}
\psfrag{D}{$\Omega^+$}
\psfrag{E}{$\bs{\Phi}$}
\psfrag{F}{$\Phi^-$}
\psfrag{G}{$\Phi^+$}
\psfrag{H}{$\Sigma$}
\includegraphics[scale=0.075]{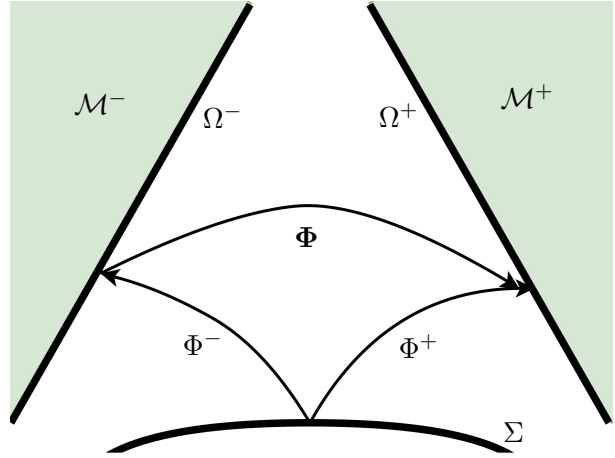}
\caption{Abstract setup of the matching of two $\lp n+1\rp-$dimensional spacetimes $\lp\Mpm,g^{\pm}\rp$ across their null boundaries $\Omega^{\pm}$, $\Sigma$ being the abstract $n-$dimensional manifold satisfying $\Phi^{\pm}\lp\Sigma\rp=\Omega^{\pm}$.}
\label{fig1}
\end{figure}

\subsection{Metric hypersurface data calculations}
A key object to study the problem is $\bs{\Phi}:=\Phi^+\circ({\Phi^-})^{-1}$. In fact, most of the information about the matching is contained in $\bs{\Phi}$ since it provides the identification between points of $\Omega^-$ and $\Omega^+$ and hence between the tangent spaces $T_{p}\Omega^-$ and $T_{\bs{\Phi}\lp p\rp}\Omega^+$. The identification of the full tangent spaces $T_p \mathcal{M}^-$ and $T_{\bs{\Phi}\lp p\rp}\mathcal{M}^+$ is performed via the identification of riggings. The embedding and the rigging on one side can be chosen freely and the matching problem requires determining the embedding and the rigging on the other side such that the metric hypersurface data on each side agree. It is important to stress \cite{mars2007lorentzian} 
that the existence of the rigging on the second side is a step that cannot be overlooked, see below.

Let $\lb \lambda,y^A\rb$ be a coordinate system on a neighbourhood of some point $q\in\Sigma$, where $\lambda$ is a coordinate along the degenerate direction of $\gamma$ and $y^A$ are spacelike coordinates. Applying the push-forward $\td\Phi^{\pm}$ to the associated coordinate vectors defines a basis  $\{ e^{\pm}_1\vert_{\Phi^{\pm}\lp q\rp}=\td \Phi^{\pm}\vert_q\lp\cp_{\lambda}\rp,e^{\pm}_I\vert_{\Phi^{\pm}\lp q\rp}=\td \Phi^{\pm}\vert_q\lp\cp_{y^I}\rp\}$ of $T_{\Phi^{\pm}\lp q\rp}\Omega^{\pm}$. Since the metric hypersurface data for $\Omega^-$ and $\Omega^+$ must coincide, (\ref{emhd}) imposes
\begin{align}
\label{sjc1}&\gamma_{ij}=\la e^-_i,e^-_j\ragm=\la e^+_i,e^+_j\ragp,\\
\label{sjc2}&\ell_{i}=\la e^-_i,\xi^-\ragm=\la e^+_i,\xi^+\ragp,\\
\label{sjc3}&\ell^{(2)}=\la\xi^-,\xi^-\ragm=\la\xi^+,\xi^+\ragp.
\end{align} 
As already mentioned, the embedding and the rigging on one side are freely specifiable. We therefore adapt $\Phi^-$ and $\xi^-$ to the geometric quantities we have introduced along $\Omega^-$, and let all the information of the matching be contained in $\Phi^+$ and $\xi^+$. Thus, without loss of generality we set
\begin{equation}
\label{eiyl}
\begin{array}{lclcl}
e^-_1=k^-, & & e^-_I=v^-_I,& &\xi^-=L^-.
\end{array}
\end{equation}
Note that this choice automatically restricts the coordinate $\lambda$. In particular, it fixes its orientation since $\lambda$ must increase to the future along the generators.

Particularizing (\ref{sjc1}) for $i=j=1$  one gets $\la e^+_1,e^+_1\ragp=0$, and hence $e^+_1$ must be proportional to the null generator $k^+$ of $\Omega^+$ and the conditions (\ref{sjc1}) for $i=1$ and $j=J$ are automatically satisfied. The vector fields $e^+_i$ and $\xi^+$ can be decomposed as
\begin{equation}
\label{eqA4}
\begin{array}{lclcl}
e^+_1=\zeta k^+, & &e^+_I=a_Ik^++b_I^Jv^+_J,&&\xi^+=\dfrac{1}{A}L^++Bk^++C^Kv^+_K,
\end{array}
\end{equation}
for suitable scalar functions $\zeta$, $a_I$, $b_I^J$, $A$, $B$ and $C^K$. In (\ref{eqA4}), we have written $1/A$ for later convenience and to emphasize that this coefficient cannot vanish (because the rigging $\xi^+$ is, by definition, transversal to $\Omega^+$). Inserting (\ref{eqA4}) and
defining $f:=\lp\nfi^-\circ\bs{\Phi}^{-1}\rp/\nfi^+$, the shell junction conditions (\ref{sjc1})-(\ref{sjc3})
take the form
\begin{align}
\label{junct1}h^-_{IJ}\vert_p &= b_I^Lb_J^Kh^+_{LK}\vert_{\bs{\Phi}\lp p\rp},\\
\label{junct2}\nfi^-\vert_p&=\dfrac{\zeta\nfi^+}{A}\Big\vert_{\bs{\Phi}\lp p\rp}\qquad\qquad\qquad\qquad\qquad\qquad\qquad\Longrightarrow\quad e^+_1=fAk^+,\\
\label{junct3}-\psi^-_I\vert_p&=-\dfrac{1}{A}\lp a_I\nfi^++b_I^J\psi^+_J\rp+C^Kb_I^Jh^+_{JK}\Big\vert_{\bs{\Phi}\lp p\rp},\\
\label{junct4}0 &= 2B\nfi^++2C^J\psi^+_J-AC^IC^Jh^+_{IJ}\Big\vert_{\bs{\Phi}\lp p\rp}.
\end{align}
From now on, we make abuse of notation by writing $\nfi^-$ instead of $\nfi^-\circ\bs{\Phi}^{-1}$ on the $\lp\Mp,g^+\rp$ side.

\begin{figure}[t]
\centering
\psfrag{M+}{$\Mp$}
\psfrag{M-}{$\Ml$}
\psfrag{O-}{$\Omega^-$}
\psfrag{O+}{$\Omega^+$}
\psfrag{L+}{$\xi^+$}
\psfrag{L}{$\xi^-$}
\psfrag{(a)}{(a)}
\psfrag{(b)}{(b)}
\psfrag{(c)}{(c)}
\psfrag{(d)}{(d)}
\includegraphics[scale=0.1]{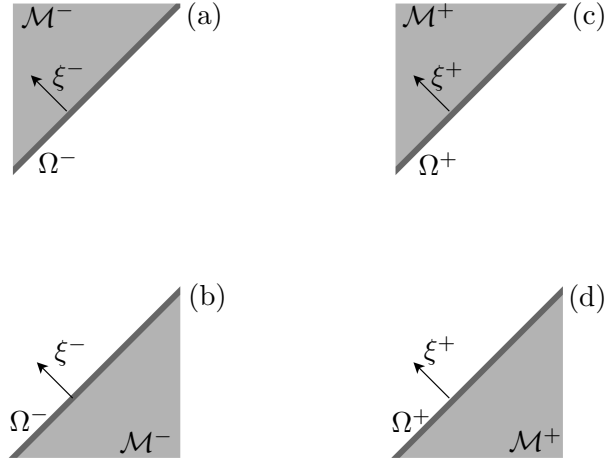}
\caption{Possible orientation of the rigging $\xi^+$ for future $\xi^-$ (see (\ref{eiyl}) and note that $L^-$ was chosen future). There exist two different situations: (a) $\Omega^-$ is past boundary and hence $\xi^-$ points inwards and (b) $\Omega^-$ is future boundary and therefore $\xi^-$ points outwards. The matching theory predicts that one rigging must point inwards and the other outwards (with respect to the spacetime boundary). Thus,  case (a) is only compatible with situation (d). Likewise, (b) can only be matched with (c). Consequently, the matching is feasible when $\Omega^{-}$ is future boundary and $\Omega^+$ is past boundary or vice versa.}
\label{figureorientations}
\end{figure}

Since in this paper we are assuming that the spacetimes to be matched are known and that the selection of the basis $\lb L^{\pm},k^{\pm},v^{\pm}_I\rb$ of $\Gamma\lp T\mathcal{M}^{\pm}\rp\vert_{\Omega^{\pm}}$ has already been made, the quantities $\psi_I^{\pm}$, $h_{IJ}^{\pm}$ and $\nfi^{\pm}$ in these expressions should be regarded as known. In such a setup, equations (\ref{junct3}) and (\ref{junct4}) allow one to solve uniquely for $B$ and $C^K$ in terms of the (still unknown) quantities $A$, $a_I$ and $b_I^K$.

The choice (\ref{eiyl}) together with (\ref{sjc3}) imposes both riggings  $\xi^{\pm}$ to be null. Given embeddings $\Phi^{\pm}$ and  $\xi^-$, equations (\ref{sjc2})-(\ref{sjc3}) admit at most one solution for $\xi^+$ (see Lemma 3 of \cite{mars2007lorentzian}). However, the existence of a solution is necessary but not quite sufficient to guarantee that the matching can be performed. The reason is that in order to be able to identify the riggings (and hence construct the matched spacetime), they both must point into the same side of the hypersurface once the matching has been carried out. This requires that the riggings must be such that if one points outwards (with respect to the spacetime boundary) then the other one must point inwards. Since there is no freedom whatsoever in the
solution of  (\ref{sjc1})-(\ref{sjc3}) when the boundaries are null\footnote{This is fundamentally different in the null and in the non-null cases, see \cite{mars2007lorentzian} and \cite{mars2020hypersurface}.} it may happen that the equations are solvable but the matching is impossible.

In order to describe when this obstruction is absent, we introduce the following terminology.
A null boundary $\Omega$ of a time-oriented spacetime $\lp\mathcal{M},g\rp$
is called \textbf{future (resp. past) boundary} if $\mathcal{M}$ lies in the past (resp. future) of $\Omega$.
In our setup we have chosen $L^-$ to be future and by \eqref{eiyl} so it is  $\xi^-$. There are only two possible situations. Either (a) $\Omega^-$ is a past boundary (and then
$\xi^-$ points inwards) or (b) $\Omega^-$ is future boundary (and $\xi^-$ points outwards), see
cases (a) and (b) in Figure \ref{figureorientations}.
In the first case, the matching can be performed only if $\xi^+$ points outwards and in the second when
$\xi^+$ points inwards. Since after the identification the causal orientation of the rigging is unique, this forces that in case (a) $\Omega^+$ must be a future boundary and in case (b)
$\Omega^+$ must be a past boundary. Using the terminology of Figure \ref{figureorientations}, the possible matchings are
(a)-(d) or (b)-(c). Since the names of the spacetimes to be matched can be swapped, we may assume without
loss of generality that the matching is of the (b)-(c) type, i.e. we assume from now that
$\Omega^-$ is a  future boundary and
$\Omega^+$ is a past boundary.

Having imposed also that $L^{\pm}, k^{\pm}$ are all future, the previous considerations imply that the matching is feasible if there exists a solution of (\ref{junct1})-(\ref{junct4}) with $A>0$ (had we selected e.g. $\xi^{-} = - L^-$, the corresponding condition to be imposed would be $A<0$).

\section{The step function $H\lp\lambda,y^A\rp$}
\label{secstepfunct}

The vector fields $e^{\pm}_i$ play a crucial role
when analysing the existence of $\bs{\Phi}$ and its determination. Let us start with basic facts. From the standard property $\td\Phi^{\pm}\lp\lc X,Y\rc\rp=\lc \td\Phi^{\pm}\lp X\rp,\td\Phi^{\pm}\lp Y\rp\rc$, $\forall X,Y \in \Gamma(T\Sigma)$, and given that $e_i^{\pm}$ are the push-forward of a coordinate basis, it must hold
\begin{equation}
\label{eqA56}
\begin{array}{lcl}
\lc e^{\pm}_i,e^{\pm}_j\rc=0.
\end{array}
\end{equation}
For $e_i^-$ these conditions give no extra information as $\{ k^-,v^-_I\}$ verify item (C) in (\ref{basis}). On the other hand, $e_i^+$ are not yet determined, so (\ref{eqA56}) provide useful information. Inserting (\ref{eqA4}) one easily finds
\begin{align}
  \label{eqA9b}0=\lc e^+_I,e^+_J\rc&=\Big( e^+_I\lp a_J\rp -e^+_J\lp a_I\rp \Big) k^++ \Big(e^+_I\lp b_J^K\rp -e^+_J\lp b_I^K\rp\Big) v^+_K, \\
\label{eqA9a} 0=\lc e^+_1,e^+_J\rc&=\Big( e^+_1\lp a_J\rp -e^+_J\lp \zeta\rp \Big) k^++e^+_1\lp b_J^I\rp v^+_I.
\end{align}
Setting each component equal to zero and using the definition of push-forward $\td \Phi^+(X) (u) = X (u \circ \Phi^+)$, we get\footnote{We make the harmless abuse of notation of calling $u \circ \Phi^+$ still as $u$.} 
\begin{align}
\label{eprop1}& e^+_I\lp a_J\rp =e^+_J\lp a_I\rp,\quad e^+_I\lp b_J^K\rp =e^+_J\lp b_I^K\rp, & &\Longleftrightarrow& &\dfrac{\cp a_J}{\cp y^I}=\dfrac{\cp a_I}{\cp y^J},\quad \dfrac{\cp b_J^K}{\cp y^I}=\dfrac{\cp b_I^K}{\cp y^J},&\\
\label{eprop2}&e^+_1\lp a_J\rp =e^+_J\lp \zeta\rp,\quad e^+_1\lp b_J^K\rp =0,& &\Longleftrightarrow&  &\dfrac{\cp a_J}{\cp\lambda}=\dfrac{\cp \zeta}{\cp y^J},\quad\dfrac{\cp b_J^K}{\cp\lambda}=0.&
\end{align}
It follows that, locally on $\Sigma$, there exist functions $H(\lambda,y^A)$ and $h^I(\lambda, y^A)$ such that 
\begin{align}
\label{eqA11a}&a_I=\dfrac{\cp H\lp \lambda, y^A\rp}{\cp y^I},& &\dfrac{\cp \zeta}{\cp y^J}=\dfrac{\cp a_J}{\cp\lambda}=\dfrac{\cp^2 H\lp \lambda, y^A\rp}{\cp\lambda\cp y^J},&\\
\label{eqA11b}&b_I^K=\dfrac{\cp h^K\lp \lambda, y^A\rp}{\cp y^I},& &b_I^K=b_I^K\lp y^A\rp.&
\end{align}
From (\ref{eqA11b}) we conclude that $h^I\lp\lambda,y^A\rp$ must decompose as
\begin{equation}
\label{setsolutions}h^{I}\lp\lambda,y^A\rp=h_{\lambda}^{I}\lp\lambda\rp+h_{y}^{I}\lp y^A\rp.
\end{equation}
The integration ``constant" $h_{\lambda}^I\lp\lambda\rp$ is irrelevant because it does not change $b_I^K$ or $a_I$, hence it affects neither $e^+_I$ nor the embedding $\Phi^+$. Thus we may set $h_{\lambda}^I\lp \lambda\rp=0$ without loss of generality and conclude $h^I=h^I\lp y^A\rp$. Concerning (\ref{eqA11a}), substitution of (\ref{junct2}) yields
\begin{equation}
  \dfrac{\cp }{\cp y^J}\lp fA-\dfrac{\cp H\lp \lambda, y^A\rp}{\cp\lambda}\rp=0, 
  \qquad \mbox{i.e.} \qquad
\dfrac{\cp H\lp\lambda, y^A\rp}{\cp\lambda}=fA+\eta\lp\lambda\rp, 
\label{solHFA}
\end{equation}
$\eta\lp\lambda\rp$ being an arbitrary function of $\lambda$. Again, $\eta\lp\lambda\rp$ plays no role since it does not affect $e^+_I$ (by (\ref{eqA4})) so we may set $\eta\lp \lambda\rp=0$. Thus, $A$ and $a_I$ can be written in terms of $H\lp\lambda,y^A\rp$ as
\begin{equation}
\label{HFA}
\dfrac{\cp H\lp\lambda,y^A\rp}{\cp\lambda}=fA,\qquad\qquad\dfrac{\cp H\lp\lambda,y^A\rp}{\cp y^I}=a_I.
\end{equation}

Since we assumed the basis $\{k^{\pm},v_I^{\pm}\}$ to satisfy conditions (\ref{basis}), the two foliations $\{ S^{\pm}_{s^{\pm}}\}$ defined by $s^{\pm}$ are such that $k^{\pm}\lp s^{\pm}\rp=1$ and $v_I^{\pm}\lp s^{\pm}\rp=0$. Consequently, it follows
\begin{align}
&e^+_1\lp s^+\rp=fAk^+\lp s^+\rp=fA=\dfrac{\cp H\lp\lambda,y^A\rp}{\cp \lambda},&  &e^-_1\lp s^-\rp=k^-\lp s^-\rp=1,&\\
&e^+_I\lp s^+\rp=a_Ik^+\lp s^+\rp=a_I=\dfrac{\cp H\lp\lambda,y^A\rp}{\cp y^I},&  &e^-_I\lp s^-\rp=v_I^-\lp s^-\rp=0,&
\end{align}
from where one concludes that the functions $s^{\pm}$ verify
\begin{equation}
\label{expforHs}
\ld\begin{array}{l}
s^-\circ\Phi^-=\lambda+\textup{const.}\\
s^+\circ\Phi^+=H+\textup{const.}
\end{array}\rb\textup{ on }\Sigma.
\end{equation}
The constants are again irrelevant and can be absorbed in
the coordinate $\lambda$ and in $H$ respectively, so we may set them to zero
without loss of generality.

Given $p^{\pm}\in\Omega^{\pm}$, the value $s^{\pm}\lp p^{\pm}\rp$ indicates at what height (as measured by
$\lambda$) the point $p^{\pm}$ is located along the null generator that contains it.
In view of (\ref{expforHs}), the function $H\lp\lambda,y^A\rp$ measures the step on the null coordinate when crossing from $\mathcal{M}^-$ to $\mathcal{M}^+$. We therefore call $H(\lambda,y^A)$ the \textbf{step function}.

This result immediately connects the cut-and-paste constructions with our formalism. In the seminal construction by Penrose \cite{dewitt1968battelle}, \cite{Penrose:1972xrn}, plane-fronted impulsive gravitational waves propagating in the Minkowski spacetime are constructed by cutting out Minkowski across a null hyperplane and reattaching the two regions after shifting the null coordinate of one of the regions. To be specific, using double null coordinates where
the Minkowski metric is $g_{\mbox{\tiny Mink}} = - 2 du dv + dx^2 + dy^2$
and the impulsive wave is located at $v=0$, the reattachment is performed
after shifting $u$ in $v=0^+$ 
by $u \rightarrow u + h(x,y)$. 
This jump is precisely of the form \eqref{expforHs} with $H = u + h(x,y)$, provided
we use $\{u = \lambda,x,y\}$ also as 
coordinates intrinsic to the null hyperplane, so that the embedding $\Phi^{-}$ becomes the identity.
Another example of the direct link between the function $H$ and the cut-and-paste construction appears in \cite{podolsky2017penrose}, where expression \eqref{expforHs} is equivalent to $H=\cv-\mathcal{H}\lp\eta,\bar{\eta}\rp$. More details about the connection between the matching formalism and the cut-and-paste construction are given in Section \ref{secPenrplane} below.

Let us pause for a moment and summarize what we have found. Assuming that
the causal orientations of the boundaries are compatible the
matching is possible if and only if
the junction conditions \eqref{junct1}-\eqref{junct4} are satisfied. The last two are always solvable and determine uniquely the coefficients $B$ and $C^I$ (i.e. the tangential components of the rigging $\xi^+$). Equation \eqref{junct2} is automatically satisfied if the embeddings $\Phi^{\pm}$ are restricted to satisfy
\eqref{expforHs} and $A$ and $a_I$ are defined by \eqref{HFA}, where
$H$  is a smooth function on $\Sigma$ satisfying $\partial_{\lambda} H \neq 0$.
Actually, with our choice that $k^{\pm}$ and $L^{\pm}$
are all future directed, the functions $\nfi^{\pm}$ are positive and so it is $f$. 
Consequently, the function $H$ must satisfy $\partial_{\lambda} H = f A >0$, or in more geometric and coordinate independent terms, that $H$ is strictly increasing
along any future null generator of $\Sigma$.

Thus, the core problem for existence of the matching is the solvability of
\eqref{junct1}.
Note that for any $p \in \Omega^-$, the constant section
$S^-_{s^-(p)} = \{ s^- = s^-(p)\} \subset \Omega^{-}$
is mapped via $\bs{\Phi}$ to the spacelike submanifold
$\bs{\Phi} (S^-_{s^-(p)} ) \subset \Omega^+$. Condition \eqref{junct1} combined with
\eqref{eqA11b} states that there exists an isometry between these two submanifolds. Even more, since $h^I$ depends only on $\{y^A\}$, this isometry must be universal in the sense of being independent of the value $s^-(p)$. This fact was already observed in \cite{blau2016horizon} (see equations (2.9)-(2.10)) and later in \cite{bhattacharjee2017soldering} when studying the coordinate changes leaving the first fundamental form $\gamma$ invariant.  
In order to describe this more explicitly, let us transfer the coordinates $\{ \lambda, y^I\}$
of $\Sigma$  to $\Omega^{-}$, so that the embedding
$\Phi^{-}$ becomes the identity map. Take now coordinates $\{ s^+, u^I \}$
on $\Omega^+$ such that
$v^+_I = \partial_{u^I}$ (in particular, they are constant along the null generators). The embedding $\Phi^+$  takes the form
\begin{align}
  \begin{array}{lcll}
    \Phi^+: & \Sigma & \longrightarrow &\Omega^+ \\
            & ( \lambda, y^I ) & \longrightarrow & \Phi^+(\lambda, y^I) =
                                   \left (s^+= H(\lambda,y^I), u^I = h^I(y^J) \right ).
  \end{array}
\label{embedPhi+}                                                   
\end{align}
       The section $S^-_{s^-_0}$ in $\Omega^-$  is mapped into $\bs{\Phi} (S^-_{s^-_0}) = \{ s^+ = H(\lambda =s^-_0, y^J), u^I(y^J) \}$.
Note that a point $p \in \Omega^-$ can be identified uniquely by specifying the
null generator to which it belongs together with its height
$s^-(p)$ along the generator, and the same happens on $\Omega^+$. Thus, 
the matching is feasible if and only if  there exists a diffeomorphism
$\Psi$  between the set of null generators of $\Omega^-$ and 
the set of null generators of $\Omega^+$ (defined locally by $u^I(y^J)$) 
such that, for each possible value of $s^-_0$,
 the map that takes each point at height $s^-_0$ along a generator
$\sigma$ in $\Omega^-$ to 
the point at height $H|_{\sigma} (s^-_0)$ in $\Omega^+$ along the  generator $\Psi(\sigma)$, happens to be an isometry.
This is of course a very strong restriction and generically it will not be
possible to find $H$ and $\Psi$ verifying it (which simply means that the matching cannot be done). However, as we see next, there are situations where the matching is not only feasible but it even allows for an infinite number of possibilities, and other cases where there is at most one possible  step function $H$
for each admissible choice of $\Psi$.

In order to describe these results, recall that
 $e^-_1=k^-$ and $e^+_1= (\partial_{\lambda} H) k^+$, so (\ref{eqA28}) immediately leads to
\begin{equation}
\label{eqA59}
\begin{array}{lcl}
e^-_1\lp h^-_{IJ}\rp=\dfrac{\cp h^-_{IJ}}{\cp\lambda}=2{\bs{\chi}}_-^{k^-}\lp v^-_I,v^-_J\rp, & &e^+_1\lp h^+_{IJ}\rp=\dfrac{\cp h^+_{IJ}}{\cp\lambda}=2fA{\bs{\chi}}_+^{k^+}\lp v^+_I,v^+_J\rp.
\end{array}
\end{equation}
The partial derivative of (\ref{junct1}) with respect to $\lambda$ gives after using (\ref{eprop2})
\begin{equation}
\label{eqA58}
\dfrac{\cp h^-_{IJ}}{\cp\lambda}=\dfrac{\cp \lp b_I^Ab_J^C\rp}{\cp\lambda}h^+_{AC}+b_I^Ab_J^C\dfrac{\cp h^+_{AC}}{\cp\lambda}=b_I^Ab_J^C\dfrac{\cp h^+_{AC}}{\cp\lambda}.
\end{equation}
Combining (\ref{eqA59}), (\ref{eqA58})  yields
\begin{equation}
\label{detrmfA}
{\bs{\chi}}_-^{k^-}\lp v^-_I,v^-_J\rp=\dfrac{\cp H\lp\lambda,y^K\rp}{\cp\lambda}b_I^Ab_J^C{\bs{\chi}}_+^{k^+}\lp v^+_A,v^+_C\rp.
\end{equation}
Since we are assuming the geometry of $\Omega^{\pm}$ to be known and that 
the choice 
of $\{ k^{\pm}, v^{\pm}_I\}$ has already been made,  this expression determines, for each possible choice of $\Psi$,
 i.e. of $b^{A}_{B}$ fulfilling \eqref{junct1},
a unique value for $\partial_{\lambda} H$ 
 unless the two second fundamental forms vanish simultaneously.  If, on the other hand, there exist open sets $\textup{ }\cu^{\pm}\subset\Omega^{\pm}$ related by  $\textup{ }\cu^+=\bs{\Phi}\lp\cu^-\rp$ and such that
\begin{equation}
\label{condicioncilla}
\begin{array}{lcl}
\bs{\chi}_-^{k^-}\lp v^-_I,v^-_J\rp\vert_{\cu^-}=0, & & {\bs{\chi}}_+^{k^+}\lp v^+_I,v^+_J\rp\vert_{\cu^+}=0,
\end{array}
\end{equation}
then \eqref{detrmfA} is identically satisfied. Under \eqref{condicioncilla},
all the spacelike sections in $\cu^{-}$ are isometric to each other, and the same happens in $\cu^{+}$ (this is a consequence of \eqref{eqA28} and the equality
\eqref{equalityhbarh}  between the quotient metric at any point $p$
and the metric of any spacelike section passing through this point). Thus, the set of null generators
can be endowed with a positive definite metric. If there is an isometry
$\Psi$ between these two spaces, then {\it any} step function
$H(\lambda,y^I)$ satisfying $\partial_{\lambda} H >0$ defines a feasible matching.
This means that a point $p \in \Omega^-$ lying on a null generator
$\sigma^-$ can be shifted arbitrarily along the null generator 
$\sigma^+ := \Psi(\sigma^{-})$ in  $\Omega^+$, with the only condition that 
if $q$ is to the future of $p$ along $\sigma^-$ then
their images have the same causal relation along 
$\sigma^+$.  The matching in these circumstances exhibits a large freedom. The results from the previous reasoning completely agree with those obtained in \cite{blau2016horizon} when studying this particular case of totally geodesic null boundaries and its associated matching freedom. Two examples of this  
are the following cut-and-paste constructions: the plane-fronted impulsive wave \cite{penrose1968twistor}, \cite{dewitt1968battelle}, \cite{Penrose:1972xrn} by
Penrose and both the non-expanding impulsive wave in constant-curvature backgrounds \cite{podolsky1999nonexpanding}, \cite{podolsky2019cut} and the impulsive wave with gyratons \cite{podolsky2017penrose} by Podolsk{\`y} and collaborators.

\section{Energy-momentum tensor of the shell}
\label{secegmom}

Let us assume that the manifolds $\lp\Mpm,g^{\pm}\rp$ are such that the conditions (\ref{sjc1})-(\ref{sjc3}) are fulfilled, so a spacetime $( \mathcal{M},g)$ containing a null shell can be constructed. Our aim in this section
is to study the energy-momentum tensor of this shell. This tensor encodes
fundamental properties of the matter-energy contents within the shell. For the computation we shall use the framework developed in \cite{mars2020hypersurface}, to which we refer for additional details.

Given any metric hypersurface data, the associated tensor $\mathcal{A}$
introduced in 
(\ref{ambientmetric}) is by definition non-degenerate and hence admits an inverse contravariant tensor ${\mathcal{A}}^{\sharp}\vert_p$, from which one can  define a symmetric $2-$contravariant tensor $P^{ab}\vert_p$, a vector field $n^a\vert_p$ and a scalar $n^{(2)}$ in $p\in\Sigma$ by means of
\begin{equation}
\label{Pnn2}
\begin{array}{lclcl}
{\mathcal{A}}^{\sharp}\vert_p\lp\lp\bs{\alpha},a\rp,\lp\bs{\beta},b\rp\rp=P\vert_p\lp\bs{\alpha},\bs{\beta}\rp+an\vert_p\lp\bs{\beta}\rp+bn\vert_p\lp\bs{\alpha}\rp+abn^{(2)}\vert_p, & & \bs{\alpha},\bs{\beta}\in T^*_p\Sigma& &a,b\in\mathbb{R}.
\end{array}
\end{equation}
As given in  Definition 10 of \cite{mars2013constraint}, the energy-momentum tensor of the shell is the symmetric $2-$contravariant tensor
\begin{equation}
\label{tau}
\tau^{ab}=\lp n^aP^{bc}+n^bP^{ac}\rp n^dV_{dc}-\lp n^{\lp2\rp}P^{ac}P^{bd}+P^{ab}n^cn^d\rp V_{cd}+\lp n^{\lp2\rp}P^{ab}-n^an^b\rp P^{cd}V_{cd},
\end{equation}
where $V_{ab}$ are the components of the tensor $\bs{V}$ introduced in  (\ref{eqA22}). In matrix notation, the tensors ${\mathcal{A}}$ and ${\mathcal{A}}^{\sharp}$ can be expressed as follows:
\begin{equation}
{\mathcal{A}}=\lp\begin{array}{cc}
\gamma_{ij} & \ell_i\\
\ell_j & \ell^{\lp2\rp}
\end{array}\rp,\qquad{\mathcal{A}}^{\sharp}=\lp\begin{array}{cc}
P^{ij} & n^i\\
n^j & n^{\lp2\rp}
\end{array}\rp.
\end{equation}
Particularizing to the metric hypersurface data obtained from (\ref{sjc1})-(\ref{sjc3}), we have $\ell^{\lp2\rp}=0$ and $\ell_1=-\nfi^-$, which gives 
\begin{equation}
\label{AyAparticular}
{\mathcal{A}}=\lp\begin{array}{ccccc}
0 & 0 & \ell_1\\
0 & \gamma_{IJ} & \ell_I\\
\ell_1 & \ell_J & \ell^{\lp2\rp}
\end{array}\rp\qquad\Longrightarrow\qquad{\mathcal{A}}^{\sharp}=\lp\begin{array}{ccccc}
P^{11} & P^{1J} & n^1\\
P^{I1} & \gamma^{IJ} & 0\\
n^1 & 0 & 0
\end{array}\rp,
\end{equation}
where $n^1=1/\ell_1=-1/\nfi^-$ and $\gamma^{IJ}$ is the inverse of
$\gamma_{IJ}=h_{IJ}^-$. Although the known tensor is actually $h_{IJ}^-$, in the following we shall use $\bs{\gamma}$ instead, so that expressions become clearer. 

Substitution of $n^a=n^1\delta_1^a$ and $n^{\lp2\rp}=0$ in (\ref{tau}) simplifies the form of the energy-momentum tensor in the present setup to be
\begin{equation}
\tau^{ab}=\lp n^1\rp^2\Big( \lp \delta^a_1P^{bc}+\delta^b_1P^{ac}\rp V_{1c}- P^{ab}V_{11} -\delta^a_1\delta^b_1 P^{cd}V_{cd}\Big),
\end{equation} 
or, in components,
\begin{equation}
\label{componentstau}
\tau^{11}=-\lp n^1\rp^2\gamma^{IJ}V_{IJ},\qquad
\tau^{1I}=\lp n^1\rp^2\gamma^{IJ}V_{1J},\qquad
\tau^{IJ}=-\lp n^1\rp^2\gamma^{IJ}V_{11}.
\end{equation}
The next step is to compute the explicit form
of the tensor $\bs{V} = \bs{Y}^+ - \bs{Y}^-$. Since the rigging has been adapted to the $\Omega^-$ side, the computation of
$\bs{Y}^+$ is considerably more involved than that of $\bs{Y}^-$. We start with a few useful lemmas that will aid us along the way.  We assume without
further notice the setup of sections 3 and 4. The Lie and exterior derivatives on $\Sigma$ are denoted by $\lieo$ and $\tdo$ and we unify notation by writing $\{\theta^1=\lambda,\theta^A=y^A\}$.
\begin{lem}
\label{lem3}
The one$-$form $\bs{L}^+:=\bs{g}^+\lp\cdot,L^+\rp$ satisfies
\begin{equation}
\label{lastlem}{\Phi^+}^*\lp \bs{L^+}\rp=-\uhat{\bs{\omega}}\qquad\textup{where}\qquad \uhat{\bs{\omega}}:=\nfi^+\tdo H+\psi_J^+\tdo h^J\in \Gamma \lp T^*\Sigma\rp.
\end{equation}
\end{lem}
\begin{proof}
For any vector field $Z\in\Gamma\lp T\Sigma\rp$, it holds
\begin{align}
\nonumber {\Phi^+}^*\lp \bs{L}^+\rp\lp Z^a\cp_{\theta^a}\rp&=\bs{L}^+\lp {\Phi^+}_*\lp Z^a\cp_{\theta^a}\rp\rp=\la L^+,Z^ae^+_a\ragp\\
\nonumber &=Z^1\la L^+,e_1^+\ragp+Z^A\la L^+,e_A^+\ragp\\
\nonumber &=Z^1\la L^+,fAk^+\ragp+Z^A\la L^+,a_Ak^++b_A^Jv_J^+\ragp\\
\nonumber &=-\nfi^+\lp fA Z^1+Z^A a_A\rp-Z^Ab_A^J\psi_J^+\\
\label{eqA60} &=-\lp\nfi^+\tdo H\lp Z\rp+\psi_J^+\tdo h^{J}\lp Z\rp\rp=-\uhat{\bs{\omega}}\lp Z\rp,
\end{align}
where we have used (\ref{junct2}), (\ref{eqA4}) in the third equality, (\ref{eqA30}) in the fourth one and (\ref{eqA11b}), (\ref{HFA}) for the last step.
\end{proof}
The rigging vector field $\xi^+$ has been decomposed in (\ref{eqA4}) into
a tangential part (the $k^+$ and $v^+_K$ components) and a transversal part (the $L^+$ component). It is convenient to introduce the
vector field $X\in\Gamma\lp T\Sigma\rp$ satisfying
\begin{align}
  \label{eqfortdX}
  \xi^+=\dfrac{1}{A}L^++Bk^++C^Kv^+_K=\dfrac{1}{A}\lp L^++\Phi^+_*\lp X\rp\rp,
 \end{align}
and to define the  functions $X^a$ on $\Omega^+$ by the decomposition
$\Phi^+_*\lp X\rp=X^1e^+_1+X^Ae^+_A$.
\begin{lem}
  \label{lem4}
  The functions $X^a$  are given by
\begin{align}
\label{XA}X^A&=\gamma^{IA}\lp \omega_I-A\psi^-_I\rp,\\
\label{X1} X^1&=- \dfrac{\gamma^{IJ}}{2\nfi^-A}\lp\omega_I-A\psi^-_I\rp\lp \omega_J+A\psi^-_J\rp.
\end{align}
Moreover, the vector field $X=X^a\cp_{\theta^a}$ satisfies
\begin{equation}
\label{gammaX}\gamma\lp X,\cdot\rp=\uhat{\bs{\omega}}-A\uhat{\bs{\psi}}^--\nfi^+\cp_{\lambda}H\tdo\lambda\qquad\textup{where}\qquad\uhat{\bs{\psi}}^-:=\psi^-_I\tdo\theta^I.
\end{equation}
\end{lem}
\begin{proof}
The shell junction conditions (\ref{sjc2}) ensure that
$\la\xi^+,e^+_a\ragp=\frac{1}{A}\lp\la L^+,e_a^+\ragp+X^b\gamma_{ab}\rp=\ell_a$ with $\ell_1=-\nfi^-$, $\ell_A=-\psi^-_A$. Using Lemma \ref{lem3} it follows
\begin{equation}
\la L^+,e^+_a\ragp=\bs{L}^+\lp e^+_a\rp=\bs{L}^+\lp \Phi^+_*\lp\cp_{\theta^a}\rp\rp={\Phi^+}^*\lp\bs{L}^+\rp\lp \cp_{\theta^a}\rp=-\uhat{\bs{\omega}}\lp\cp_{\theta^a}\rp=-\omega_a.
\end{equation}
Consequently,
\begin{equation}
\label{Xeq1}
X^b\gamma_{ab}=\omega_a+A\ell_a,
\end{equation}
which proves both (\ref{gammaX}) and (\ref{XA}) after using
\begin{align}
  \nfi^-A = \nfi^+fA= \nfi^+\cp_{\lambda}H. \label{AderH}
\end{align}
Condition (\ref{sjc3}) and the fact that both $\xi^- = L^{-}$
and $L^+$ are null give
\begin{equation}
\label{Xeq2}
0=\la\xi^+,\xi^+\ragp=\dfrac{1}{A^2}\lp 2X^a\la L^+,e^+_a\ragp+X^aX^b\gamma_{ab}\rp\quad\Longrightarrow\quad -2X^a\omega_a+X^aX^b\gamma_{ab}=0.
\end{equation}
Combining this  with (\ref{Xeq1}) and using $X^a\ell_a=-X^1\nfi^--X^A\psi^-_A$ yields
\begin{equation}
\label{Xeq3}-X^a\omega_a=A\lp X^1\nfi^-+X^A\psi^-_A\rp\qquad\Longrightarrow\qquad 2X^1 \nfi^+\cp_{\lambda}H=-X^A\lp\omega_A+A\psi^-_A\rp,
\end{equation}
which gives (\ref{X1}) after using again \eqref{AderH}.
\end{proof}

\begin{cor}
\label{cor1}
In the basis $\{L^+,k^+,v_I^+\}$ of $\Gamma\big( T\Mp\big)\vert_{\Omega^+}$, the rigging $\xi^+$ can be expressed as
\begin{equation}
\label{uniqxi}
\xi^+=\dfrac{\nfi^-}{\cp_{\lambda}H}\lp \frac{1}{\nfi^+}
  L^++  h_+^{AB}\lp (b^{-1})^I_A \lp \cp_{y^I}H
  - \frac{1}{\nfi^-} \cp_{\lambda} H \psi^-_I \rp + \frac{1}{\nfi^+}
  \psi^+_A   \rp Z_B\rp,
\end{equation}
where $Z_B:=  \frac{1}{2} \lp (b^{-1})_B^J \lp
  \cp_{y^J}H -  \frac{1}{\nfi^{-}} \cp_{\lambda}H \psi_J^- \rp -\frac{1}{\nfi^+}
  \psi^+_B   \rp k^++v^+_B$.
\end{cor}
\begin{proof}
From the first junction condition \eqref{junct1} one deduces
\begin{align}
\nonumber &\delta^{A}_C=h_-^{AB}h_{BC}^-\stackbin{\eqref{junct1}}{=}h_-^{AB}b_B^Ib_C^Jh_{IJ}^+\quad\Longrightarrow\quad (b^{-1})^A_K=h_-^{AB}b_B^Ih_{IK}^+\\
\label{invsh+}&\quad\Longrightarrow\quad  h_+^{KL} (b^{-1})_{K}^A=h_-^{AB}b_B^L\quad\Longrightarrow\quad h_+^{KL}(b^{-1})_{K}^A(b^{-1})_L^J=h_-^{AJ}\equiv\gamma^{AJ}.
\end{align}	
The shell junction condition \eqref{junct2} together with \eqref{HFA} give
  \begin{align*}
    e^+_1=\cp_{\lambda}Hk^+, \qquad e^+_I=\cp_{y^I}Hk^++b_I^Jv^+_J.
    \end{align*}
  The result \eqref{uniqxi} follows from \eqref{eqfortdX} after inserting
  \eqref{XA}-\eqref{X1} and using  the definition
  \eqref{lastlem} of $\uhat{\bs{\omega}}$ and \eqref{invsh+}.
\end{proof}
For the sake of simplicity, we introduce the notation
\begin{align}
  \label{hatquantities}
\hat{\psi}_I^+:=b_I^K\psi^+_K=\psi^+_K\cp_{y^I}h^K, \qquad
\hat{v}^+_I:=b_I^Kv_K^+,
\end{align}
and similarly for other objects carrying capital Latin indices.
\begin{lem}
\label{lem5}
The following identities hold:
\begin{align}
\label{derA}\dfrac{\cp_{\theta^a}A}{A}=&\textup{ }\dfrac{\cp_{\theta^a}\cp_{\lambda}H}{\cp_{\lambda}H}+ \dfrac{\cp_{\theta^a}\nfi^+}{\nfi^+}-\dfrac{\cp_{\theta^a}\nfi^-}{\nfi^-},\\
\label{porsi1} \la\nabla^+_{e^+_1}L^+,e_1^+\ragp=&\textup{ }\cp_{\lambda}H\lp \nfi^+\kappa_{k^+}^+\cp_{\lambda}H-\cp_{\lambda}\nfi^+ \rp,\\
\nonumber \la\nabla^+_{e^+_1}L^+,e_J^+\ragp+\la\nabla^+_{e^+_J}L^+,e_1^+\ragp=&\textup{ }\cp_{\lambda}H\bigg( 2\nfi^+\kappa_{k^+}^+\cp_{y^J}H-\dfrac{\cp_{\lambda}\hat{\psi}^+_J}{\cp_{\lambda}H} -2\nfi^+\bs{\sigma}_{L^+}^+\lp \hat{v}_J^+\rp-\cp_{y^J}\nfi^+ \bigg)\\
\label{porsi2} &\textup{ }-\cp_{y^J}H\textup{ }\cp_{\lambda}\nfi^+,\\
\nonumber \la\nabla^+_{e^+_I}L^+,e^+_J\ragp=&\textup{ }\nfi^+\kappa_{k^+}^+\cp_{y^I}H\textup{ }\cp_{y^J}H-\dfrac{\cp_{y^I}H\textup{ }\cp_{\lambda}\hat{\psi}^+_J}{\cp_{\lambda}H}-\cp_{y^J}H\textup{ }\cp_{y^I}\nfi^+\\
\label{porsi3} &\textup{ }-\nfi^+\lp\cp_{y^I}H\textup{ }\bs{\sigma}_{L^+}^+\lp\hat{v}^+_J\rp+\cp_{y^J}H\textup{ }\bs{\sigma}_{L^+}^+\lp \hat{v}_I^+\rp\rp +\bs{\Theta}^{L^+}_+\lp\hat{v}_I^+,\hat{v}_J^+\rp.
\end{align}
\end{lem}
\begin{proof}
   We shall use repeatedly the decompositions
  $e_1^+=\cp_{\lambda}H\textup{ }k^+$, $e^+_I=\cp_{y^I}H\textup{ }k^++\hat{v}_I^+$
  which follow directly from  (\ref{eqA4})-(\ref{junct2}) and (\ref{HFA}).
  For the first claim in the lemma, we compute 
 \begin{align}
\nonumber \dfrac{\cp_{\theta^a}A}{A}=\dfrac{1}{A}\cp_{\theta^a}\lp\dfrac{fA}{f}\rp=\dfrac{1}{fA}\lp\cp_{\theta^a}\cp_{\lambda}H-A\cp_{\theta^a}f\rp= \dfrac{\cp_{\theta^a}\cp_{\lambda}H}{\cp_{\lambda}H}-\dfrac{\cp_{\theta^a}f}{f},
  \end{align}
  which leads to  (\ref{derA}) by simply inserting $f:=\nfi^-/\nfi^+$. For the second expression, we use (\ref{derkL}) so that
\begin{equation}
\la\nabla^+_{e^+_1}L^+,e_1^+\ragp=\lp\cp_{\lambda}H\rp^2\la\nabla^+_{k^+}L^+,k^+\ragp=\lp\cp_{\lambda}H\rp^2 \lp \nfi^+\kappa^+_{k^+}-k^+\lp\nfi^+\rp \rp
\end{equation}
which can be rewritten as (\ref{porsi1}). For the third expression we compute each term in the left-hand side separately. In both cases
we use the covariant derivatives of $L$ given in Lemma \ref{lem1}. Firstly,
\begin{align}
\nonumber \la\nabla^+_{e^+_1}L^+,e_J^+\ragp&=\cp_{\lambda}H \la\nabla^+_{k^+}L^+,\cp_{y^J}H\textup{ }k^+ + \hat{v}_J^+\ragp\\
\nonumber &=\cp_{y^J}H\lp \nfi^+\kappa_{k^+}^+\cp_{\lambda}H-\cp_{\lambda}\nfi^+ \rp-\cp_{\lambda}H\lp k^+\big( \hat{\psi}_J^+\big)+\nfi^+\bs{\sigma}_{L^{+}}^+\lp \hat{v}_J^+\rp\rp\\
\label{l4ref1} &=\cp_{y^J}H\lp \nfi^+\kappa_{k^+}^+\cp_{\lambda}H-\cp_{\lambda}\nfi^+ \rp- \cp_{\lambda}\hat{\psi}_J^+-\nfi^+\cp_{\lambda}H\textup{ }\bs{\sigma}_{L^{+}}^+\lp \hat{v}_J^+\rp,
\end{align}
where in the second equality we used $k^+\lp b_I^J\rp=0$. Secondly,
\begin{align}
\nonumber \la\nabla^+_{e^+_J}L^+,e_1^+\ragp&=\cp_{\lambda}H \lp \cp_{y^J}H\la\nabla^+_{k^+}L^+,k^+\ragp+\la\nabla^+_{\hat{v}^+_J}L^+,k^+\ragp\rp\\
\nonumber &=\cp_{y^J}H\lp \nfi^+\kappa_{k^+}^+\cp_{\lambda}H-\cp_{\lambda}\nfi^+ \rp-\cp_{\lambda}H\lp \hat{v}_J^+\lp \nfi^+\rp+\nfi^+\bs{\sigma}^+_{L^+}\lp\hat{v}_J^+\rp\rp\\
\nonumber &=\cp_{y^J}H\lp \nfi^+\kappa_{k^+}^+\cp_{\lambda}H-\cp_{\lambda}\nfi^+ \rp-\cp_{\lambda}H\lp \lp e_J^+-\dfrac{\cp_{y^J}H}{\cp_{\lambda}H} e^+_1\rp \lp \nfi^+\rp+\nfi^+\bs{\sigma}^+_{L^+}\lp\hat{v}_J^+\rp\rp\\
\label{l4ref2} &=\cp_{\lambda}H\Big( \nfi^+\kappa_{k^+}^+\cp_{y^J}H- \cp_{y^J}\nfi^+ -\nfi^+\bs{\sigma}^+_{L^+}\lp\hat{v}_J^+\rp\Big),
\end{align}
which immediately leads to (\ref{porsi2}). Finally, for the term $\la\nabla^+_{e^+_I}L^+,e^+_J\ragp$ one obtains
\begin{align}
\nonumber\la\nabla^+_{e^+_I}L^+,e^+_J\ragp=&\textup{ }\dfrac{\cp_{y^I}H}{\cp_{\lambda}H}\lp \cp_{y^J}H\lp \nfi^+\kappa_{k^+}^+\cp_{\lambda}H-\cp_{\lambda}\nfi^+\rp-\cp_{\lambda}\hat{\psi}^+_J-\nfi^+\cp_{\lambda}H\textup{ }\bs{\sigma}_{L^+}^+\lp\hat{v}^+_J\rp\rp\\
\nonumber &\textup{ }+\cp_{y^J}H\textup{ }\la\nabla^+_{\hat{v}_I^+}L^+,k^+\ragp +\la \nabla^+_{\hat{v}_I^+}L^+,\hat{v}_J^+\ragp\\
\nonumber =&\textup{ }\dfrac{\cp_{y^I}H}{\cp_{\lambda}H}\lp  \nfi^+\kappa_{k^+}^+\cp_{y^J}H\textup{ }\cp_{\lambda}H-\cp_{\lambda}\hat{\psi}^+_J-\nfi^+\cp_{\lambda}H\textup{ }\bs{\sigma}_{L^+}^+\lp\hat{v}^+_J\rp\rp\\
\nonumber &\textup{ }-\cp_{y^J}H\lp \cp_{y^I}\nfi^++\nfi^+\bs{\sigma}_{L^+}^+\lp \hat{v}_I^+\rp \rp +\bs{\Theta}^{L^+}_+\lp\hat{v}_I^+,\hat{v}_J^+\rp,
\end{align}
from where  (\ref{porsi3}) follows at once.
\end{proof}

The tensors $\bs{Y}^{\pm}$ are directly defined in terms of the pull-backs
${\Phi^{\pm}}^*\lp \mathcal{L}_{\xi^{\pm}}g^{\pm}\rp$. We
compute ${\Phi^{+}}^*\lp \mathcal{L}_{\xi^{+}}g^{+}\rp$ and then get
${\Phi^{-}}^*\lp \mathcal{L}_{\xi^{-}}g^{-}\rp$ as a suitable specialization. The computation relies on the following fundamental relationship between Lie derivatives and embeddings.
Let $\Omega$ be an embedded hypersurface on $\lp\mathcal{M},g\rp$ with embedding $\Phi:\Sigma\hookrightarrow\Omega$ and first fundamental form $\gamma$. Then, given a scalar function $\rho:\Omega\longrightarrow\mathbb{R}$ and a vector field $Y\in\Gamma\lp T\Sigma\rp$, the following identity holds
\begin{equation}
\label{lieprop}
\Phi^*\lp\mathcal{L}_{\rho\textup{ }\Phi_*\lp Y\rp}g\rp=\rho\lieo_Y\gamma+\tdo\rho\otimes\gamma\lp Y,\cdot\rp+\gamma\lp Y,\cdot\rp\otimes\tdo\rho.
\end{equation}
For the transversal part of  ${\Phi^{+}}^*\lp \mathcal{L}_{\xi^{+}}g^{+}\rp$
we shall use
\begin{align}
\label{pbL1}{\Phi^{\pm}}^*\lp \mathcal{L}_{L^{\pm}}g^{\pm}\rp=&\lp\la\nabla^{\pm}_{e^{\pm}_a}L^{\pm},e^{\pm}_b\ragp+\la\nabla^{\pm}_{e^{\pm}_b}L^{\pm},e^{\pm}_a\ragp\rp\tdo{\theta}^a\otimes\tdo{\theta}^b.
\end{align}
From the decomposition (\ref{eqfortdX}) we get
\begin{align}
\nonumber {\Phi^+}^*&\lp\mathcal{L}_{\xi^+}g^+\rp=\textup{ }{\Phi^+}^*\lp\mathcal{L}_{\frac{1}{A}\lp L^++\Phi^+_*\lp X\rp\rp}g^+\rp\\
\nonumber =&\textup{ }{\Phi^+}^*\lp \frac{1}{A}\mathcal{L}_{\lp L^++\Phi^+_*\lp X\rp\rp}g^+-\dfrac{\td A}{A^2}\otimes g^+\lp L^++\Phi^+_*\lp X\rp,\cdot\rp-g^+\lp L^++\Phi^+_*\lp X\rp,\cdot\rp\otimes \dfrac{\td A}{A^2} \rp\\
  \nonumber  \stackrel{\eqref{lieprop}}{=}                                           &\textup{ }\dfrac{1}{A}{\Phi^+}^*\lp \mathcal{L}_{ L^+}g^+\rp+\dfrac{1}{A}\lieo_X\gamma-\dfrac{\tdo A}{A^2}\otimes \lp {\Phi^+}^*\lp \bs{L}^+\rp+\gamma\lp X,\cdot\rp\rp-\lp {\Phi^+}^*\lp \bs{L}^+\rp+\gamma\lp X,\cdot\rp\rp\otimes \dfrac{\tdo A}{A^2}\\
  \nonumber =&\textup{ }\dfrac{1}{A}\lp{\Phi^+}^*\lp \mathcal{L}_{ L^+}g^+\rp+\lieo_X\gamma+\dfrac{\tdo A}{A}\otimes \lp A\uhat{\bs{\psi}}^-+\nfi^+\cp_{\lambda}H\tdo\lambda\rp+\lp A\uhat{\bs{\psi}}^-+\nfi^+\cp_{\lambda}H\tdo\lambda\rp\otimes \dfrac{\tdo A}{A}\rp\\
\label{pullbackxi} =&\textup{ }\dfrac{1}{A}\lp{\Phi^+}^*\lp \mathcal{L}_{ L^+}g^+\rp+\lieo_X\gamma+\tdo A\otimes \uhat{\bs{\psi}}^-+\uhat{\bs{\psi}}^-\otimes \tdo A+\nfi^{-}\lp\tdo A\otimes\tdo\lambda+\tdo\lambda\otimes\tdo A\rp\rp,
\end{align}
where Lemma \ref{lem3} and (\ref{gammaX}) are used in the fourth equality
and \eqref{AderH} in the last one.
We next compute $\lieo_{X} \gamma$.  Since the first fundamental form $\gamma$ is degenerate, i.e. $\gamma_{1A}=0$, one gets
\begin{equation}
\label{liegamma}\lp\lieo_{X}\gamma\rp_{ab}=X^c\cp_{\theta^c}\gamma_{ab}+\gamma_{aI}\cp_{\theta^b}X^{I}+\gamma_{Ib}\cp_{\theta^a}X^{I}.
\end{equation}  
Denoting the Lie derivative and Levi-Civita connection on a section $\lambda=\textup{const.}$ of $\Sigma$ by $\lieo^{\parallel}$ and $\nabla^{\parallel}$ respectively and inserting (\ref{XA})-(\ref{X1}) into (\ref{liegamma}) yields
\begin{align}
\label{complie11}\lp\lieo_{X}\gamma\rp_{11}&=0,\\
\label{complie1J}\lp\lieo_{X}\gamma\rp_{1J}&=\gamma_{JL}\cp_{\lambda}X^L=\cp_{\lambda}\lp \gamma_{JL}X^L\rp-X^L\cp_{\lambda}\gamma_{JL}=\cp_{\lambda}\lp\omega_J-A\psi^-_J\rp-2X^L\bs{\chi}_-^{k^-}\lp v_J^-,v_L^-\rp,\\
\nonumber \lp\lieo_{X}\gamma\rp_{IJ}&=X^1\cp_{\lambda}\gamma_{IJ}+X^L\cp_{y^L}\gamma_{IJ}+\gamma_{IL}\cp_{y^J}X^{L}+\gamma_{LJ}\cp_{y^I}X^{L}\\
\label{complieIJ}&=2X^1\bs{\chi}_-^{k^-}\lp v_I^-,v_J^-\rp+\lieo^{\parallel}_{X}\gamma_{IJ}=2X^1\bs{\chi}_-^{k^-}\lp v_I^-,v_J^-\rp+\nabla^{\parallel}_IX_J+\nabla^{\parallel}_JX_I,\textcolor{white}{-----}
\end{align}
where $X_I:=\gamma_{IL}X^L$. By (\ref{XA}) and (\ref{lastlem}), the covariant derivative $\nabla^{\parallel}_IX_J$ can be expanded  to
\begin{align}
\nonumber \nabla_I^{\parallel}X_J&=\nabla_I^{\parallel}\lp \omega_J-A\psi^-_J\rp=\nabla_I^{\parallel}\lp \nfi^+\nabla_J^{\parallel}H+\hat{\psi}^+_J-A\psi^-_J\rp\\
\label{covX} &=\nabla_I^{\parallel}\nfi^+\nabla_J^{\parallel}H+\nfi^+\nabla_I^{\parallel}\nabla_J^{\parallel}H+\nabla_I^{\parallel}\hat{\psi}^+_J-A\nabla_I^{\parallel}\psi^-_J-\psi^-_J\nabla_{I}^{\parallel}A.
\end{align}

We have now all the ingredients to compute $\bs{Y}^{\pm}$ and the energy-momentum tensor on the shell. The result  is given in the next proposition (where brackets, as usual, denote symmetrization).
\begin{prop}
\label{prop6} The tensor $\bs{Y}^{+}$ has the following components:
\begin{align}
\label{Y11}Y_{11}^+=&\textup{ } \nfi^{-}\lp \kappa_{k^+}^+\cp_{\lambda}H + \dfrac{\cp_{\lambda}\cp_{\lambda}H}{\cp_{\lambda}H}-\dfrac{\cp_{\lambda}\nfi^-}{\nfi^-}\rp,\\
\label{Y1J} Y_{1J}^+=&\textup{ }  \nfi^-\Bigg( \kappa_{k^+}^+\nabla_{J}^{\parallel}H-\bs{\sigma}_{L^+}^+\lp \hat{v}_J^+\rp+\dfrac{\cp_{\lambda}\cp_{y^J}H}{\cp_{\lambda}H}-\dfrac{X^L\bs{\chi}_-^{k^-}\lp v_J^-,v_L^-\rp}{\nfi^+\cp_{\lambda}H}-\dfrac{\nabla_J^{\parallel}\nfi^-}{2\nfi^-}-\dfrac{\cp_{\lambda}\psi_J^-}{2\nfi^-}\Bigg),\\
\nonumber Y^+_{IJ}=&\textup{ }\nfi^-\Bigg(\dfrac{\kappa_{k^+}^+\nabla_{I}^{\parallel}H \textup{ }\nabla_{J}^{\parallel}H}{\cp_{\lambda}H}-\dfrac{\nabla_{{\lp I\rd}}^{\parallel}H\textup{ }\cp_{\lambda}\hat{\psi}^+_{\ld J\rp}}{\nfi^+\lp\cp_{\lambda}H\rp^2}-\dfrac{2\nabla_{{\lp I\rd}}^{\parallel}H\textup{ }\bs{\sigma}_{L^+}^+\lp\hat{v}^+_{\ld J\rp}\rp}{\cp_{\lambda}H} +\dfrac{\bs{\Theta}^{L^+}_+\lp\hat{v}_{\lp I\rd}^+,\hat{v}_{\ld J\rp}^+\rp}{\nfi^+\cp_{\lambda}H}\\
\label{YIJ}&\textup{ }+ \dfrac{X^1\bs{\chi}_-^{k^-}\lp v_I^-,v_J^-\rp}{\nfi^+\cp_{\lambda}H}+\dfrac{\nabla_{I}^{\parallel}\nabla_{J}^{\parallel}H}{\cp_{\lambda}H}+\dfrac{\nabla_{\lp I \rd}^{\parallel}\hat{\psi}^+_{\ld J \rp}}{\nfi^+\cp_{\lambda}H}-\dfrac{\nabla_{\lp I \rd}^{\parallel}\psi^-_{\ld J \rp}}{\nfi^-}\Bigg),
\end{align}
while  $\bs{Y}^{-}$ is
\begin{align}
\label{Ymenos}
Y^{-}_{11} = \nfi^-\lp\kappa^-_{k^-}-\dfrac{\cp_{\lambda}\nfi^-}{\nfi^-}\rp,\quad Y^{-}_{1J} =-\nfi^-\lp\bs{\sigma}^-_{L^-}\lp v^-_J\rp+ \dfrac{\nabla^{\parallel}_{J}\nfi^-}{2\nfi^-}+\dfrac{\cp_{\lambda}\psi^-_J}{2\nfi^-}  \rp,\quad Y^{-}_{IJ} = \bs{\Theta}^{L^-}_-\lp v_{\lp I \rd}^-,v_{\ld J \rp}^-\rp.
\end{align}
Consequently, the components of the energy-momentum tensor of the shell are given by
\begin{align}
\nonumber \tau^{11}=&\textup{ }-\dfrac{\gamma^{IJ}}{\nfi^-}\Bigg( \dfrac{\kappa_{k^+}^+\nabla_{I}^{\parallel}H \textup{ }\nabla_{J}^{\parallel}H}{\cp_{\lambda}H}-\dfrac{\nabla_{{I}}^{\parallel}H\textup{ }\cp_{\lambda}\hat{\psi}^+_{J}}{\nfi^+\lp\cp_{\lambda}H\rp^2}-\dfrac{2\nabla_{{I}}^{\parallel}H\textup{ }\bs{\sigma}_{L^+}^+\lp\hat{v}^+_{J}\rp }{\cp_{\lambda}H}+\dfrac{\bs{\Theta}^{L^+}_+\lp\hat{v}_{I}^+,\hat{v}_{J}^+\rp}{\nfi^+\cp_{\lambda}H}\\
\label{finaltau1} &\textup{ }+ \dfrac{X^1\bs{\chi}_-^{k^-}\lp v_I^-,v_J^-\rp}{\nfi^+\cp_{\lambda}H}+\dfrac{\nabla_{I}^{\parallel}\nabla_{J}^{\parallel}H}{\cp_{\lambda}H}+\dfrac{\nabla_{I}^{\parallel}\hat{\psi}^+_{J}}{\nfi^+\cp_{\lambda}H}-\dfrac{\nabla_{I}^{\parallel}\psi^-_{J}}{\nfi^-}- \dfrac{\bs{\Theta}^{L^-}_-\lp v_{I}^-,v_{J}^-\rp}{\nfi^-}\Bigg),\\
\label{finaltau2} \tau^{1I}=&\textup{ }\dfrac{\gamma^{IJ}}{{\nfi^-}}\Bigg( \kappa_{k^+}^+\nabla_{J}^{\parallel}H+\dfrac{\cp_{\lambda}\cp_{y^J}H}{\cp_{\lambda}H}-\dfrac{X^L\bs{\chi}_-^{k^-}\lp v_J^-,v_L^-\rp}{\nfi^+\cp_{\lambda}H}-\lp\bs{\sigma}_{L^+}^+\lp \hat{v}_J^+\rp-\bs{\sigma}^-_{L^-}\lp v^-_J\rp\rp \Bigg),\\
\label{finaltau3} \tau^{IJ}=&\textup{ }-\dfrac{\gamma^{IJ}}{{\nfi^-}}\lp  \kappa_{k^+}^+\cp_{\lambda}H-\kappa^-_{k^-} + \dfrac{\cp_{\lambda}\cp_{\lambda}H}{\cp_{\lambda}H} \rp.
\end{align}
\end{prop}
\begin{proof}
Using (\ref{pullbackxi}) and the definition of $\bs{Y}^+$ one finds
\begin{align}
\label{afuodsn} \bs{Y}^+=&\textup{ }\dfrac{1}{2A}\lp{\Phi^+}^*\lp \mathcal{L}_{ L^+}g^+\rp+\lieo_X\gamma+\tdo A\otimes \uhat{\bs{\psi}}^-+\uhat{\bs{\psi}}^-\otimes \tdo A+\nfi^{-}\lp\tdo A\otimes\tdo\lambda+\tdo\lambda\otimes\tdo A\rp\rp,
\end{align}
For the $Y^+_{11}$ component, substitution of
(\ref{complie11}), (\ref{pbL1}) and (\ref{porsi1}) yields
\begin{align}
Y^+_{11}=&\textup{ }\dfrac{1}{2A}\lp{\Phi^+}^*\lp \mathcal{L}_{ L^+}g^+\rp_{11}+2\nfi^{-}\cp_{\lambda}A\rp=\nfi^{-}\lp \kappa_{k^+}^+\cp_{\lambda}H-\dfrac{\cp_{\lambda}\nfi^+}{\nfi^+} +\dfrac{\nfi^-}{\nfi^+}\dfrac{\cp_{\lambda}A}{\cp_{\lambda}H}\rp,
\end{align}
which is (\ref{Y11}) after replacing $\cp_{\lambda}A$ as given in
(\ref{derA}). Similarly, (\ref{afuodsn}) together with the definition of $\uhat{\bs{\omega}}$ gives
\begin{align}
\nonumber Y_{1J}^+=&\textup{ }\dfrac{1}{2A}\lp{\Phi^+}^*\lp \mathcal{L}_{ L^+}g^+\rp_{1J}+\lp\lieo_X\gamma\rp_{1J}+\psi^-_J\cp_{\lambda} A +\nfi^{-}\cp_{y^J} A\rp\\
\nonumber =&\textup{ }\dfrac{1}{2A}\Bigg( \cp_{\lambda}H\bigg( 2\nfi^+\kappa_{k^+}^+\nabla_{J}^{\parallel}H-2\nfi^+\bs{\sigma}_{L^+}^+\lp \hat{v}_J^+\rp-\nabla_{J}^{\parallel}\nfi^+ \bigg)\\
\nonumber &\textup{ }+\nfi^+\cp_{\lambda}\cp_{y^J}H-A\cp_{\lambda}\psi_J^--2X^L\bs{\chi}_-^{k^-}\lp v_J^-,v_L^-\rp +\nfi^{-}\nabla_{J}^{\parallel} A\Bigg),
\end{align}
from where (\ref{Y1J}) is deduced after inserting (\ref{derA}).
For the last set of components, we combine (\ref{afuodsn}) with (\ref{pbL1}), (\ref{complieIJ}), (\ref{covX}) and (\ref{porsi3}) to get
\begin{align}
\nonumber Y_{IJ}^+=&\textup{ }\dfrac{1}{2A}\lp{\Phi^+}^*\lp \mathcal{L}_{ L^+}g^+\rp_{IJ}+\lp\lieo_X\gamma\rp_{IJ}+\psi_I^- \cp_{y^J} A+\psi_J^-\cp_{y^I} A \rp\\
\nonumber =&\textup{ }\dfrac{1}{A}\lp \la\nabla^{+}_{e^{+}_{\lp I\rd}}L^{+},e^{+}_{\ld J\rp}\ragp+ X^1\bs{\chi}_-^{k^-}\lp v_I^-,v_J^-\rp+\nabla^{\parallel}_{\lp I\rd}X_{\ld J\rp}+\psi_{\lp J\rd}^- \nabla_{\ld I\rp}^{\parallel}A\rp\\
\nonumber =&\textup{ }\dfrac{1}{A}\lp \la\nabla^{+}_{e^{+}_{\lp I\rd}}L^{+},e^{+}_{\ld J\rp}\ragp+ X^1\bs{\chi}_-^{k^-}\lp v_I^-,v_J^-\rp+  \nabla^{\parallel}_{\lp I\rd}\nfi^+\nabla^{\parallel}_{\ld J\rp}H+\nfi^+\nabla^{\parallel}_{\lp I\rd}\nabla^{\parallel}_{\ld J\rp}H+\nabla^{\parallel}_{\lp I\rd}\hat{\psi}^+_{\ld J\rp}-A\nabla_{\lp I\rd}^{\parallel}\psi^-_{\ld J\rp}\rp\\
\nonumber =&\textup{ }\dfrac{1}{A}\bigg( \nfi^+\kappa_{k^+}^+\nabla_{{\lp I\rd}}^{\parallel}H \textup{ }\nabla_{{\ld J\rp}}^{\parallel}H-\dfrac{\nabla_{{\lp I\rd}}^{\parallel}H\textup{ }\cp_{\lambda}\hat{\psi}^+_{\ld J\rp}}{\cp_{\lambda}H}-2\nfi^+\nabla_{{\lp I\rd}}^{\parallel}H\bs{\sigma}_{L^+}^+\lp\hat{v}^+_{\ld J\rp}\rp +\bs{\Theta}^{L^+}_+\lp\hat{v}_{\lp I\rd}^+,\hat{v}_{\ld J\rp}^+\rp\\
\nonumber &\textup{ }+ X^1\bs{\chi}_-^{k^-}\lp v_I^-,v_J^-\rp+\nfi^+\nabla_{\lp I \rd}^{\parallel}\nabla_{\ld J \rp}^{\parallel}H+\nabla_{\lp I \rd}^{\parallel}\hat{\psi}^+_{\ld J \rp}-A\nabla_{\lp I \rd}^{\parallel}\psi^-_{\ld J \rp}\bigg),
\end{align}
which  becomes (\ref{YIJ}) upon using (\ref{HFA}) and $\nabla_{\lp I\rd}^{\parallel}\nabla_{\ld J\rp}^{\parallel}H=\nabla_I^{\parallel}\nabla_J^{\parallel}H$.
To get $\bs{Y}^-$ it suffices to particularize (\ref{Y11})-(\ref{YIJ}) for $b_I^J=\delta_I^J$, $X^a=0$ and $H\lp\lambda,y^A\rp=\lambda$, as well as replacing all $+$ superscripts by $-$.


The components of the energy-momentum tensor are obtained from
(\ref{componentstau}) by direct subtraction of the explicit
expressions of $Y_{ab}^-$ and $Y_{ab}^+$ and using $n^1=1/\ell_1=-1/\nfi^-$, c.f. (\ref{AyAparticular}).
\end{proof}
It is a general fact of the geometry of shells (see Proposition 7 in \cite{mars2013constraint}) that
the energy-momentum tensor on the shell depends on the choice of rigging solely
by scale. More precisely, let $\tau^{ab}$ be the energy-momentum tensor
associated to a choice of rigging $\xi$ and $\widetilde{\tau}^{ab}$ the
 energy-momentum tensor of the same shell with respect to a different
choice of rigging $\widetilde{\xi}$. Then, decomposing (uniquely) $\widetilde{\xi}$
as  $\widetilde{\xi} = u \xi + T$, with $T$ tangent to the matching hypersurface, the energy-momentum tensors are related by
$\widetilde{\tau}^{ab} = u^{-1} \tau^{ab}$.  This fact can be used to perform
a non-trivial
consistency check on the expressions \eqref{finaltau1}-\eqref{finaltau3}. Indeed, we may choose any other null transverse vector
$\widetilde{L}^-=\alpha L^-+\beta k^-+q^Iv_I^-$ (with
$\beta$ and $q^I$ suitably restricted to preserve the null character
of $L^-$) and introduce all the geometric expressions defined in terms of
$\widetilde{L}^-$. Then the corresponding expression for 
$\widetilde{\tau}^{ab}$ can be proved to satisfy $\widetilde{\tau}^{ab}= \alpha^{-1} \tau^{ab}$, as required. See appendix A for details in this regard.

We emphasize that the energy-momentum tensor on the shell depends strongly on the step function $H$. However,  there is also dependence on the
map $\Psi$ sending null generators to null generators. This dependence is encoded in the hatted quantities $\hat{\psi}^+_I$ and $\hat{v}^+_I$ introduced in \eqref{hatquantities}.

In $8\pi G=c=1$ units, the different components of the energy-momentum tensor can be interpreted physically as an energy density
$\rho := \tau^{11}$, energy-flux $j^A := \tau^{1A}$ and pressure $p$ such that $\tau^{AB}=p\gamma^{AB}$, see e.g. \cite{poisson2004relativist}.

For later use, we recall that the energy-momentum tensor on a shell satisfies the Israel equations (also called {\it shell equations} or {\it surface layer equations}) which
in the null case were first obtained by Barrab{\'e}s and Israel
\cite{barrabes1991thin}. In the framework of hypersurface data, they may be written as \cite{mars2013constraint}
\begin{align}
\label{sfe1}\dfrac{1}{\sqrt{\vert\det\mathcal{A}\vert}}\cp_{\theta^a}\lp \sqrt{\vert\det\mathcal{A}\vert}\tau^{ab}\ell_b\rp-\dfrac{1}{2}\tau^{ab}\lp Y_{ab}^++Y_{ab}^-\rp=&\lc\rho_{\ell}\rc,\\
\label{sfe2}\dfrac{1}{\sqrt{\vert\det \mathcal{A}\vert}}\cp_{\theta^b}\lp \sqrt{\vert\det \mathcal{A}\vert}\tau^{bc}\gamma_{ca}\rp-\dfrac{1}{2}\tau^{bd}\cp_{\theta^a}\gamma_{bd}=&\lc J_a\rc,
\end{align}
where $\lc\rho_{\ell}\rc:=\rho_{\ell}^+-\rho_{\ell}^-$, $\lc J_a\rc:=J_a^+-J_a^-$, and
the bulk energy and momentum quantities $\rho^{\pm}_{\ell}$, $J_a^{\pm}$ are defined by (we correct two sign typos in Definition 9 of \cite{mars2013constraint}) 
\begin{align}
\label{Jrho}\rho_{\ell}^{\pm}=-{\Phi^{\pm}}^*\lp G^{\pm}\lp\xi^{\pm},\nu^{\pm}\rp\rp ,\qquad \bs{J}^{\pm}=-{\Phi^{\pm}}^*\lp G^{\pm}\lp\cdot,\nu^{\pm}\rp\rp.
\end{align}
Here $G^{\pm}$ is the Einstein tensor of $\lp\Mpm,g^{\pm}\rp$ and $\nu^{\pm}$
is the normal vector to the hypersurface normalized to $\la \xi^{\pm}, \nu^{\pm}
\ra_{g^{\pm}}=1$.

\section{Penrose's case: plane-fronted impulsive wave}
\label{secPenrplane}

One remarkable benefit of using the previous formalism is that multiple sorts of matchings can be analysed at once (for instance by considering a family of energy-momentum tensors satisfying \eqref{sfe1}-\eqref{sfe2}, or a set of functions $H\lp\lambda,y^A\rp$ with certain properties). This task becomes significantly more difficult by means of the cut-and-paste method. A great amount of interesting matchings have been studied with the latter, which makes impossible for us to cover all of them. In what follows, we shall use our formalism to analyse the first cut-and-paste construction, namely the plane-fronted impulsive wave (see the works \cite{penrose1965remarkable}, \cite{penrose1968twistor}, \cite{dewitt1968battelle}, \cite{Penrose:1972xrn} by Penrose). We will recover the results from cut-and-paste and obtain some new shells by posing different setups of energy, energy flux and pressure. 

As previously mentioned, the cut-and-paste approach arises with the publications \cite{dewitt1968battelle}, \cite{Penrose:1972xrn} by Penrose. The starting point is the plane-fronted wave, with well-known metric (see e.g. \cite{brinkmann1923riemann}, \cite{witten1962gravitation})
\begin{equation}
\label{rightmetric}
\td s^2=-2\lp \td \cv+\Psi\lp \cu,x,z\rp\td\cu\rp\td\cu+\td x^2+\td z^2.
\end{equation}
The spacetimes describing purely gravitational waves, i.e. solutions of the vacuum Einstein field equations, correspond to $\lp\frac{\cp^2}{\cp x^2}+\frac{\cp^2}{\cp z^2}\rp\Psi=0$. 
Penrose addresses the impulsive case of \eqref{rightmetric} by setting $\Psi\lp \cu,x,z\rp$ to zero except on the hypersurface defined by $\cu=0$, i.e. $\Psi\lp \cu,x,z\rp=\delta\lp\cu\rp\mathcal{H}\lp x,z\rp$ (where $\delta$ denotes Dirac delta function and $\mathcal{H}\lp x,z\rp$ is any real function). Under these circumstances, the metric becomes
\begin{equation}
\label{metric61}\td s^2=-2\lp \td \cv+\delta\lp\cu\rp\mathcal{H}\lp x,z\rp\td\cu\rp\td\cu+\td x^2+\td z^2.
\end{equation}

The possibility to perform a coordinate change which turns \eqref{metric61} into a $C^0$ form is already mentioned by Penrose \cite{dewitt1968battelle}, \cite{Penrose:1972xrn}. In fact, by writing \eqref{metric61} in terms of the coordinates $\{\cu,\cv,\e:=\frac{1}{\sqrt{2}}\lp x+i z\rp,\be:=\frac{1}{\sqrt{2}}\lp x-i z\rp\}$, which yields $\td s^2=-2\lp \td \cv+\delta\lp\cu\rp\mathcal{H}\lp \e,\be\rp\td\cu\rp\td\cu+2\td\e\td\be$, Podols{\'y} et al.\ \cite{podolsky1999nonexpanding}, \cite{podolsky2017penrose} find the suitable coordinate transformation
\begin{align}
\label{coordtransf}
&\cu=U,&&\cv=V+\Theta h+U_+ h,_{Z}h,_{\bar{Z}},&&\eta=Z+U_+ h,_{\bar{Z}},&
\end{align}
where the comma denotes partial derivative, $\Theta\lp U\rp$ is the Heaviside step function, $U_{+}:=U\Theta\lp U\rp$ is the so-called kink function and $h\lp Z,\bar{Z}\rp:=\mathcal{H}\lp\e,\be\rp\vert_{\cu=0}$ is a real-valued function. Inserting \eqref{coordtransf} into \eqref{metric61}, one obtains the following continuous metric\footnote{As pointed out in \cite{podolsky2017penrose}, to obtain \eqref{metricPodolsky} one needs to use $\frac{\td\Theta}{\td U}=\delta$, $\frac{\td \lp U\Theta\rp}{\td U}=\Theta$, $\Theta^2=\Theta$, which in general may lead to mathematical inconsistencies.}:
\begin{equation}
\label{metricPodolsky}\td s^2=2\lv \td Z+U_{+}\lp h,_{\bar{Z}Z}\td Z+h,_{\bar{Z}\bar{Z}}\td \bar{Z}\rp\rv^2-2\td U\td V.
\end{equation}

The transformation \eqref{coordtransf} immediately shows that the lightlike coordinate $\cv$ is discontinuous across the hypersurface $\cu=0$ and that the presence of the $\delta$-function on \eqref{metric61} is due to this jump. More precisely, the discontinuous coordinates $\{\cu,\cv,\e,\be\}$, chosen to preserve the Minkowski form of \eqref{metric61} on $\cu\gtrless0$, produce discontinuities on the metric, while with the continuous coordinates $\{U,V,Z,\bar{Z}\}$ the metric tensor becomes $C^0$ but loses the Minkowski form for $U>0$. Nevertheless, as we show next, the coordinates $\{\cu,\cv,\e,\be\}$ are useful to understand this spacetime as the outcome of the disjoint union of $\cu>0$ and $\cu<0$ with a jump on $\cv$ when crossing the hypersurface $\cu=0$.

When applying the cut-and-paste method to plane-fronted impulsive waves, Penrose proposes a jump on the lightlike coordinate $\cv$ of the form $\cv_+\vert_{\cu_+=0}=\cv_-+\mathcal{H}\lp x_-,z_-\rp\vert_{\cu_-=0}$, where $\{\cv_{\pm},x_{\pm},z_{\pm}\}$ refer to the coordinates $\{\cv,x,z\}$ on the regions $\cu \gtrless 0$ of \eqref{metric61} respectively. This jump follows directly from the coordinate transformation \eqref{coordtransf}. 

Let us therefore consider two spacetimes $\lp\Mpm,g^{\pm}\rp$, respectively corresponding to the $\cu \gtrless 0$ Minkowski regions and with metrics
$\td s_{\pm}^2=-2 \td \cv_{\pm}\td\cu_{\pm}+\delta_{AB} dx_{\pm}^A dx_{\pm}^{B}$. Since the construction below applies to all dimensions we let $A,B = 2,\ldots, n$
  so that the spacetimes are $n+1$ dimensional. The matching hypersurface is $\Omega^{\pm} = \{ \cu_{\pm} =0\}$ and we take
  $\{k^{\pm}=\cp_{\cv_{\pm}},v^{\pm}_I=\cp_{x^I_{\pm}}\}$ as a basis of $\Gamma\lp T\Omega^{\pm}\rp$ and $s^{\pm} = \cv_{\pm}$ as the foliation defining function.  These objects clearly  satisfy (\ref{basis}). As transverse null vector we select $L^{\pm}=\cp_{\cu_{\pm}}$
  (note that both $L^{\pm}$ and $k^{\pm}$ are future). With these choices it is straightforward that  
  $\nfi^{\pm}=1$ and $\psi_{I}^{\pm}$, $\kappa_{k^{\pm}}^{\pm}$, $\bs{\sigma}_{L^{\pm}}^{\pm}\lp v^{\pm}_I\rp$, $\bs{\chi}^{k^{\pm}}_{\pm}\lp v^{\pm}_I,v^{\pm}_J\rp$, $\bs{\Theta}^{L^{\pm}}_{\pm}\lp v^{\pm}_I,v^{\pm}_J\rp$ all vanish.

To make $\Phi^{-}$ the identity embedding to $\Omega^-$ we let $\{\theta^i \} = \{ \lambda,y^I\}$ be defined by $\{ \lambda, y^I\} =  \{\cv_{-}, x^I_{-}\} |_{\Omega^{-}}$. By
\eqref{embedPhi+}, we know that the embedding $\Phi^+$ must take the form
\begin{align*}
  \Phi^+(\lambda,y^I ) = \{ \cu_{+} =0, \cv_{+} = H(\lambda, y^I), x_{+}^I = h^I (y^J)\}
  \end{align*}
We observe that Penrose's jump corresponds to a step function of the form $H\lp\lambda,y^A\rp=\lambda+\mathcal{H}\lp y^A\rp$. The matching of two Minkowski regions with this $H\lp\lambda,y^A\rp$ will therefore result in spacetimes describing plane-fronted impulsive waves (purely gravitational when $\mathcal{H}\lp y^A\rp$ is harmonic). The framework introduced in previous sections, however, must provide all the possible matchings and hence a more general set of step functions. Let us analyze this general matching problem and
  discuss some interesting particular cases.

The shell junction conditions \eqref{junct1}-\eqref{junct4} impose
\begin{align}
\label{jcMkcase}
  \delta_{IJ} = b_I^Lb_J^K\delta_{LK}, \qquad b^{I}_J := \frac{\partial h^I}{\partial y^J}
\end{align}
and determine uniquely the rigging $\xi^{+}$, which from Corollary \ref{cor1} takes the explicit form
\begin{align*}
  \xi^+= \frac{1}{\partial_{\lambda} H} \left (
  L^+ +  \delta^{AB} (b^{-1})^I_A  \partial_{y^I} H  \left ( \frac{1}{2}  (b^{-1})^J_B \partial_{y^J} H   k^+
  + v^+_B \right ) \right ),
\end{align*}
whereas the metric hypersurface data is $\{\gamma_{ab}=\delta_a^A\delta_b^B\delta_{AB},\ell_a=-\delta_a^1,\ell^{(2)}=0\}$. As discussed above, \eqref{jcMkcase} constitutes an isometry condition between the spatial sections of $\Omega^{\pm}$, defined in the present case by $\cv_{\pm}=\textup{const}$ and which are simply euclidean planes. The corresponding isometries are obviously translations and rotations. This freedom in the matching could be absorbed in a rotation and translation in the coordinates $\{\cu_+,\cv_+,x_+,z_+\}$ which would set $b_I^J=\delta_I^J$. However, the results that follow are insensitive to $b_I^J$ so we avoid applying this transformation in the $\lp\Mp,g^+\rp$ side.

Particularizing the results of Proposition \ref{prop6}, one easily finds 
\begin{equation}
\label{YpmMk}
Y_{ab}^-=0,\qquad Y_{ab}^{+}=\dfrac{\cp_{\theta^a}\cp_{\theta^b}H}{\cp_{\lambda}H},
\end{equation}
whereas the energy-momentum tensor of the shell turns out to be
\begin{align}
\label{tauMk}
& \tau^{11}=-\dfrac{\delta^{IJ}\nabla^{\parallel}_I  \nabla^{\parallel}_JH}{\cp_{\lambda}H}=-\dfrac{\delta^{IJ}\cp_{y^I} \cp_{y^J}H}{\cp_{\lambda}H}, && \tau^{1I}=\dfrac{\delta^{IJ}\cp_{\lambda} \cp_{y^J}H}{\cp_{\lambda}H} , && \tau^{IJ}=-\dfrac{\delta^{IJ}\cp_{\lambda} \cp_{\lambda}H}{\cp_{\lambda}H}. &
\end{align}
Thus, $\delta_{IJ}\tau^{IJ}=-\frac{\lp n-1\rp\cp_{\lambda}\cp_{\lambda}H}{\cp_{\lambda}H}$ and $\textup{ }\delta_{IJ}\tau^{1I}=\frac{\cp_{\lambda}\cp_{y^J}H}{\cp_{\lambda}H}$, which can be combined with expressions \eqref{YpmMk}-\eqref{tauMk} to obtain
\begin{align}
\label{tauYm}
\tau^{ab}Y^+_{ab}=\dfrac{2\delta^{IJ}}{\lp\cp_{\lambda}H\rp^2}\lp \cp_{\lambda}\cp_{y^I}H\textup{ }\cp_{\lambda}\cp_{
y^J}H-\cp_{y^I}\cp_{y^J}H\textup{ }\cp_{\lambda}\cp_{\lambda}H\rp=2\delta_{IJ}\lp \tau^{1I}\tau^{1J}-\dfrac{\tau^{11}\tau^{IJ}}{n-1}\rp.
\end{align}
This, together with $\bs{Y}^-=0$, $\vert\det\mathcal{A}\vert=1$, $\tau^{ab}\ell_b=-\tau^{1a}$, $\tau^{ba}\gamma_{ac}=\delta_c^B\gamma_{AB}\tau^{bA}$ and the vanishing of the Einstein tensor in Minkowski, brings the shell field equations \eqref{sfe1}-\eqref{sfe2} into the following form
\begin{align}
\label{Mksfe}& 0=\cp_{\lambda}\tau^{11}+\cp_{y^A}\tau^{1A}+\delta_{IJ}\lp \tau^{1I}\tau^{1J}-\dfrac{\tau^{11}\tau^{IJ}}{n-1}\rp,&& 0=\cp_{\lambda}\tau^{1A}+\cp_{y^B}\tau^{BA}.&&
\end{align}

Let us start by considering no shell, i.e. $\bs{V}=0$. The results for this case should be viewed as a consistency check, since the absence of shell must give rise to the whole Minkowski spacetime. Integrating the equations \eqref{YpmMk} with $Y^+_{ab}=0$ yields $H\lp\lambda,y^A\rp=a\lambda+c_Jy^J+d$, where $a,c_J,d\in\mathbb{R}$ and $a>0$. Due to the trivial form of the embedding $\Phi^-$ and the fact that $s^+=\cv_+=H\lp\lambda,y^A\rp$, this step function corresponds to the jump
$\cv_+=a\cv_-+ c_J y^J +d$ when crossing the hypersurface $\cu_{\pm}=0$. This means that the only possible isometries between the boundaries $\Omega^{\pm}$ are (besides the translations and
  rotations in the $\{x^I_+\}$ coordinates already discussed) null translations and null rotations in the $\lp\Mp,g^+\rp$ side . Since all of them are isometries of the Minkowski metric, the matching indeed recovers the global Minkowski spacetime.

We next consider the vacuum case, i.e. $\tau^{ab}=0$. Integrating \eqref{tauMk} with the l.h.s.\ equal to zero gives the step function 
\begin{align}
\label{stepfuncstar}
&\tau^{1J}=\tau^{IJ}=0 && \Longleftrightarrow  && H\lp\lambda,y^A\rp=a\lambda+\mathcal{H}\lp y^A\rp,&\\
\label{stepfuncstar2}
&\tau^{11}=0 && \Longleftrightarrow && \sum_{I=2}^3\frac{\cp^2 \mathcal{H}}{\lp\cp y^I\rp^2}=0\quad\textup{and}\quad a>0, a\in\mathbb{R}.&
\end{align}
The freedom in $a$ corresponds to a boost in the $\lp\mathcal{M}^+,g^{+}\rp$ spacetime so we may set $a=1$ without loss of generality and hence recover Penrose's step. Note that setting $\tau^{ab}=0$ automatically forces $\mathcal{H}\lp y^A\rp$ to be harmonic, which is consistent with the Dirac delta limit of $\Psi\lp\cu,x,z\rp$ when the vacuum equations for \eqref{rightmetric} are imposed.

As a simple generalization of this example, one can consider non-zero energy, i.e. $\tau^{11}\neq0$, while keeping $\tau^{1J}=\tau^{IJ}=0$. This does not change the form of the step function, which is still given by \eqref{stepfuncstar}. It follows that Penrose's step function \eqref{stepfuncstar} corresponds to absence of pressure and energy flux. Therefore, it describes either purely gravitational waves (vacuum case) or shells of null dust, i.e. a pressureless fluid of massless particles moving at the speed of light. Now, (\ref{Mksfe}) implies that $\tau^{11}$ must be $\lambda$-independent. Writing $\tau^{11}=\rho\lp y^A\rp$ and using \eqref{tauMk} yields
\begin{equation}
\sum_{I=2}^3\dfrac{\cp^2 \mathcal{H}}{\lp\cp y^I\rp^2}=-a\rho\lp y^A\rp.
\end{equation}
Again, the constant $a$ can be set to one by applying a boost in $\lp\Mp,g^+\rp$. Observe that the energy condition $\rho\lp y^A\rp\geq0$ is equivalent to $H\lp\lambda,y^A\rp$ being a superharmonic function.

Finally, let us keep both the energy and the energy flux of the shell completely free and consider a non-zero pressure $p\lp\lambda,y^A\rp$. This case is no longer included in Penrose's cut-and-paste constructions. Since $\cp_{\lambda}H>0$, the pressure can be expressed as $p=-\cp_{\lambda}\lp \ln\lp\cp_{\lambda}H\rp\rp$, whose integration gives $\cp_{\lambda}H=\alpha\lp y^A\rp\exp\lp-\int p\lp\lambda,y^A\rp\td\lambda\rp$, where $\alpha\lp y^A\rp>0$ is the integration ``constant". Therefore,
\begin{equation}
\label{MkHpress}
H\lp\lambda,y^A\rp=\alpha\lp y^A\rp\int\exp\lp-\int p\lp\lambda,y^A\rp\td\lambda\rp\td\lambda+\mathcal{H}\lp y^A\rp,
\end{equation}
where $\mathcal{H}\lp y^A\rp$ is a second integration function.

In order to discuss the effect of the pressure in the matching, we start by noting the following simple consequences of $e^-_1 = k^-$
  and $e^+_1 = (\partial_{\lambda} H) k^+$  combined with the fact that $k^{\pm}$ are geodesic and affinely parametrized
\begin{equation}
\label{spmMk} 
\begin{array}{l}
e^-_1\lp s^-\rp=1,\\
e^+_1\lp s^+\rp=\cp_{\lambda}H,
\end{array}\quad\qquad \begin{array}{l}
\nabla_{e^-_1}^-e_{1}^-\lp s^-\rp=0,\\
\nabla_{e^+_1}^+e_{1}^+\lp s^+\rp=\cp_{\lambda}\cp_{\lambda}H.
\end{array}
\end{equation}

Consider two null generators $\sigma^-\subset\Omega^-$, $\sigma^+=\bs{\Phi}\lp\sigma^-\rp\subset\Omega^+$. Both functions $s^{\pm}$ have been built so that their rate of change measured by $k^{\pm}$ is equal to one, c.f. \eqref{basis}. We call ``velocity''  the rate of change of $s^{\pm}$ along a null vector
  along $\Omega^{\pm}$ and   ``acceleration'' the rate of change of the velocity. 
The matching, however, does not identify the vectors $k^{\pm}$ but the vectors $e^{\pm}_1$. Therefore, when moving along $\sigma^{\pm}\subset\Omega^{\pm}$, the velocity and acceleration associated to $e^{\pm}_1$ (i.e. as measured by $\lambda$) can be different, see \eqref{spmMk}. Let us hence take $\lambda$ as the measure parameter for both sides. This allows us to introduce the concepts of self-compression and self-stretching of points along any null generator $\sigma^{\pm}$. \textit{There will exist self-compression (resp. self-stretching) whenever the acceleration measured by $\lambda$ is strictly negative (resp. positive)}. Accordingly, this effect will not take place on $\Omega^-$ due to its identification with $\Sigma$, but it may certainly occur in $\Omega^{+}$. Equations \eqref{spmMk} show that the velocity and the acceleration are respectively given by the first and second derivatives of $H\lp\lambda,y^A\rp$. Consequently, this effect is ruled by the pressure, as it essentially determines $\cp_{\lambda}\cp_{\lambda}H$ at each point $q\in\Sigma$. Note that vanishing pressure entails constant velocity, which obviously gives no self-compression nor self-stretching. However, the velocity along the curves $\sigma^{\pm}$ can still be different (this is why we are not using the terms ``stretching'' or ``compressing'', which would still be occurring in this situation).  From the definition of the pressure, it follows that $\textup{sign}\lp p\lp\lambda,y^A\rp\rp=-\textup{sign}\lp \cp_{\lambda}\cp_{\lambda}H\rp$. Consequently, if the pressure is positive (resp. negative) (c.f. \eqref{spmMk}), then the acceleration along $e_1^+$  is negative (resp. positive) and there exists self-compression (resp. self-stretching) of points towards the future. Alternatively, one can conclude that \textit{positive pressure pushes points towards lower values of $H\lp \lambda,y^A\rp$ (or $s^+$) and vice versa}.

For a better understanding of this behaviour, let us consider a pressure depending only on $\lambda$ and write $p\lp\lambda\rp=-\frac{\mu''}{\mu'}$, where $'$ denotes derivative with respect to $\lambda$ and $\mu\lp\lambda\rp$ is any regular function with $\mu'\lp\lambda\rp>0$ $\forall\lambda$.  
From \eqref{MkHpress}, it follows that  $\cp_{\lambda}H=\alpha\lp y^A\rp\mu'\lp\lambda\rp>0$ and $H\lp\lambda,y^A\rp=\alpha\lp y^A\rp\mu\lp\lambda\rp+\mathcal{H}\lp y^A\rp$, after simple redefinitions of $\alpha(y^A)$ and $\mathcal{H}(y^A)$. Note that necessary and sufficient conditions for the range of the embedding
  $\Phi^+$ to be the whole of $\Omega^+$ is that $\alpha(y) >0$ and that
$\lim_{\lambda\rightarrow\pm\infty}\mu\lp\lambda\rp=\pm\infty$ (recall that $\mu(\lambda)$ is monotonically increasing).
The components of the energy-momentum tensor are 
\begin{align}
\label{egshellmu}
&  \tau^{11}=-\dfrac{1}{\alpha \mu'} \delta^{IJ} \lp\mu \cp_{y^I} \cp_{y^J}\alpha+ \cp_{y^I} \cp_{y^J}\mathcal{H}\rp, && \tau^{1I}=\dfrac{1}{\alpha }\delta^{IJ} \cp_{y^J}\alpha , && \tau^{IJ}=-\delta^{IJ} \dfrac{\mu''}{\mu'}.&
\end{align}
Observe that this setup is still fairly general in the sense that it allows for energy-momentum tensors with all components different from zero. The specific behaviour of the energy-momentum tensor is obviously ruled by $\mu\lp\lambda\rp$ and the particular form of the functions $\alpha\lp y^A\rp, \mathcal{H}\lp y^A\rp$. It is now clear that fixing the pressure amounts to setting the form of $H\lp\lambda,y^A\rp$, which contains the information about the effect of self-compression or self-stretching on the $\Omega^+$ boundary. 

As an example, let us define the function $\nu\lp\lambda\rp:=\sqrt{(a+2)\lambda^2+b^2}$ and consider 
\begin{equation}
\label{examplepressmu}
\mu\lp\lambda\rp:= \lp a +1\rp\lambda-\sqrt{a}\nu\lp\lambda\rp
\end{equation}
with $a>0$ and $b$ real constants. As the inequality $a+1-\sqrt{a (a+2)}>0$ holds for all positive $a$, this function satisfies
$\lim_{\lambda\rightarrow\pm\infty}\mu\lp\lambda\rp=\pm\infty$. The previous expressions yield
\begin{align}
&\mu'\lp\lambda\rp=a+1-\dfrac{\sqrt{a} (a+2) \lambda}{\nu\lp\lambda\rp}>0, \qquad \mu''\lp\lambda\rp=-\dfrac{\sqrt{a} (a+2) b^2}{\nu^{3}\lp\lambda\rp}\leq0&\\
\label{pressureex}&p\lp\lambda\rp=\dfrac{\sqrt{a} (a+2) b^2}{\nu^{2}\lp\lambda\rp \left(\lp a+1\rp \nu\lp\lambda\rp-\sqrt{a}(a+2)  \lambda\right)}\geq0,&\\
\label{stepfex}&H\lp\lambda,y^A\rp=\alpha\lp y^A\rp\lp\lp a +1\rp\lambda-\sqrt{a}\nu\lp\lambda\rp\rp+\mathcal{H}\lp y^A\rp, &
\end{align}
and energy density of the shell is given by
\begin{align}
\label{egformu}
& \tau^{11}=-\dfrac{\nu}{\alpha}\dfrac{\lp\lp a+1\rp\lambda-\sqrt{a}\nu\rp \delta^{IJ} \cp_{y^I} \cp_{y^J}\alpha+ \delta^{IJ} \cp_{y^I} \cp_{y^J}\mathcal{H}
}{ \lp a+1\rp\nu-\sqrt{a}\lp a+2\rp\lambda}.
\end{align}
This density diverges asymptotically at infinity (i.e. for  $\lambda\rightarrow\pm\infty$) unless $\alpha\lp y^A\rp$ is harmonic. If $b$ vanishes we have zero pressure
and we fall into a previous case ($H$ linear in $\lambda$). When $b \neq 0$, the pressure is everywhere regular, positive and
vanishes asymptotically at infinity. 
\begin{figure}[t]
\centering
\psfrag{H}{$H\lp\lambda,y^A=0\rp$}
\psfrag{p}{$p\lp\lambda\rp$}
\psfrag{x}{\hspace{-0.23cm}$\lambda$}
\psfrag{m}{Slope $1$ line}
\psfrag{s}{}
\psfrag{eg}{$\tau^{11}\lp\lambda\rp$}
\psfrag{C}{}
\psfrag{L}{}
\includegraphics[scale=1]{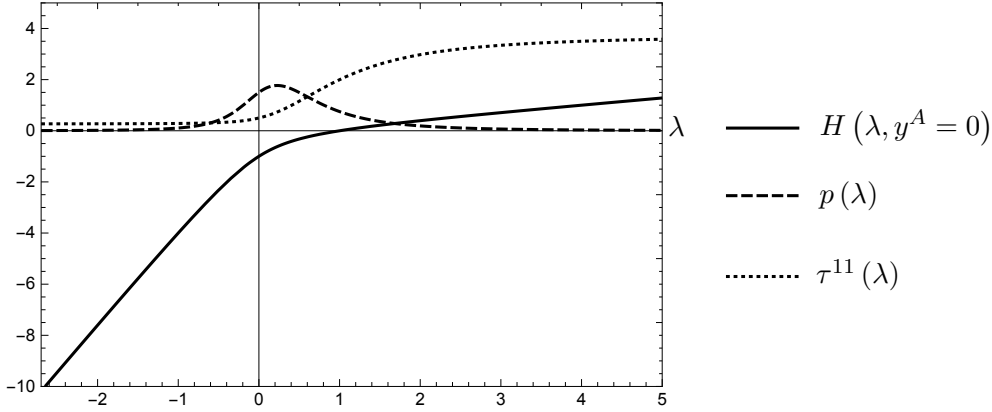}
\caption{Matching of the two Minkowski regions $\cu<0$, $\cu>0$: plot of the pressure,  step function, and  energy density of the shell
    along the null generator $\{ y^A =0\}$  for the particular values $a=1$, $b=1$, $\alpha (y^A)=1$ and  $\mathcal{H} (y^A) =\frac{1}{2(n-1)} \delta_{IJ} y^I y^J$.}
 \label{figMkpress}
\end{figure}
Under the restriction that $\alpha(y^I)$ is harmonic, a typical plot of $p\lp\lambda\rp$, $H\lp\lambda,y^A\rp$ and $\tau^{11}\lp\lambda,y^A\rp$ along a null generator of $\Omega^+$ is depicted in Figure \ref{figMkpress}.
For large negative values of $\lambda$, the step function exhibits a straight line behaviour which is a consequence of the fact that
  the pressure is negligibly small at past infinity. When $p\lp\lambda\rp$ starts increasing, the self-compression of points starts taking place and this forces the slope of $H\lp\lambda,y^A\rp$ to decrease until it reaches  again an almost constant value in the late future, once  the pressure becomes again negligible. The growth of the energy begins when the self-compression occurs and ends when the pressure approaches next-to-zero values. It tends to a finite positive value when the pressure vanishes, which suggests that \textit{it only increases (resp. decreases) on regions where there exists self-compression (resp. self-stretching)}, showing an accumulative behaviour.

To illustrate that not all the choices for the pressure result in successful matchings, we consider one last case: positive constant pressure $p$ (the negative case is completely analogous). Then, integrating \eqref{MkHpress} yields $H\lp\lambda,y^A\rp=\mathcal{H}\lp y^A\rp-\frac{1}{p}\alpha\lp y^A\rp e^{-p\lambda}$ and hence $\cp_{\lambda}\cp_{\lambda}H=-p\alpha\lp y^A\rp e^{-p\lambda}<0$. From \eqref{tauMk}, the energy and energy flux of the shell are 
  \begin{align*}
    \tau^{11} = \dfrac{\delta^{IJ}}{\alpha}\lp \dfrac{1}{p}\cp_{y^I}\cp_{y^J}\alpha -  e^{p \lambda} \cp_{y^I}\cp_{y^J}\mathcal{H}\rp, \quad \quad
    \tau^{1 I} = \delta^{IJ} \frac{\partial_J \alpha}{\alpha}.
  \end{align*}
In this situation, one finds that $\lim_{\lambda\rightarrow+\infty}H\lp\lambda,y^A\rp=\mathcal{H}\lp y^A\rp$. The positive pressure produces sustained and systematic
self-compression of points for all values of $\lambda$, which eventually results in a positive upper bound for the step function. This spoils the matching, as all the points $p\in\Omega^+$ with $s^+\lp p^+\rp>\mathcal{H}\lp y^A\rp$ cannot be identified with any  point of $\Omega^-$ or, in other words, the hypersurface $\Omega^-$ is mapped onto the proper subset
$\{ s^+ < \mathcal{H} \} \subset \Omega^-$ 
via $\bs{\Phi}$.

This last example suggests that finding possible matchings with non-zero pressure may be a significantly complicated task, specially in non-Minkowski spacetimes. In any case, the influence of the pressure producing a kind of {\it self-compression/self-stretching} of  points along the matching and its associated
  energy storage  is an interesting effect that, in our opinion, deserves
further investigation.

\appendix 

\section{Change of $\xi^-$: behaviour of the energy-momentum tensor}

As mentioned in section \ref{secegmom}, given a vector field $T$ tangent to the matching hypersurface and two rigging vector fields $\xi$, $\widetilde{\xi}$ related by $\widetilde{\xi} = u \xi + T$, their corresponding shell energy-momentum tensors satisfy $\widetilde{\tau}^{ab} = u^{-1} \tau^{ab}$ \cite{mars2013constraint}. As a consistency test, let us check that the energy-momentum tensor \eqref{finaltau1}-\eqref{finaltau3} behaves in this way. 

Let us assume that a matching of two spacetimes $\lp\Mpm,g^{\pm}\rp$ has been performed and that the rigging has been fixed by (\ref{eiyl}) after
a selection of future null transverse field $L^-$. We may repeat the matching process using a different future null transverse field $\widetilde{L}^-=\alpha L^-+\beta k^-+q^Iv_I^-$ with corresponding rigging vector $\wt{\xi}^-=\wt{L}^-$. Using tilde for all objects constructed with $\wt{L}$ the definitions \ref{eqA30} imply
\begin{equation}
\label{nfipsiapB}
\wt{\nfi}^-=\alpha\nfi^-,\qquad\wt{\psi}_I^-=\alpha\psi_I^--q_I,
\end{equation}
while the null character of $\wt{L}^-$ imposes 
\begin{equation}
\label{ap2}
2\alpha\lp \beta\nfi^-+q^I\psi_I^-\rp=\qhm\qquad\textup{where}\qquad\qhm:=q^Iq^Jh^-_{IJ}.
\end{equation}

Changing the rigging on the $\lp\Ml,g^{-}\rp$ side keeps the vector fields $e^{\pm}_a$ invariant (they only depend on the embeddings). This means that $\zeta$, $a_I$ and $b_I^J$ do not change either. On the other hand, the identification of the riggings of both sides implies that the rigging in the $\lp\Mp,g^+\rp$ side also gets modified. Let us decompose it as $\wt{\xi}^+=(1/\wt{A})\big( L^++\wt{X}^ae^+_a\big)$. Since $\zeta=\frac{A\nfi^-}{\nfi^+}$, the shell junction condition (\ref{junct2}) forces
\begin{equation}
\zeta=\dfrac{\wt{A}\wt{\nfi}^-}{\nfi^+}=\dfrac{\wt{A}\alpha\nfi^-}{\nfi^+}\qquad\Longleftrightarrow\qquad\wt{A}\alpha=A.
\end{equation}

Recalling that $\xi^+=(1/A)\lp L^++X^ae^+_a\rp$, we observe that the two riggings $\wt{\xi}^+$ and $\xi^+$ are related via $\wt{\xi}^+=\alpha\xi^++\frac{\alpha}{A}\big( \wt{X}^a-X^a\big) e^+_a$. Inserting this relation into $\la \wt{\xi}^+,e^+_B\ragp=-\wt{\psi}_B^-$, $\la\wt{\xi}^+,\wt{\xi}^+\ragp=0$ and using \eqref{nfipsiapB} and \eqref{ap2} yields
\begin{equation}
\label{ares1} \wt{X}^B=X^B+\frac{A}{\alpha}q^B,\qquad \wt{X}^1=X^1+\frac{ A\beta}{\alpha}.
\end{equation}

Each component of the energy-momentum tensor (c.f.\ \eqref{finaltau1}-\eqref{finaltau3}) is multiplied by $1/\wt{\nfi}^-$. Therefore, the transformation law of the energy-momentum tensor will be guaranteed provided each bracket in \eqref{finaltau1}-\eqref{finaltau3} turns out to be invariant. The only parts that are not trivially invariant are 
\begin{align}
\label{bracket1}& \gamma^{IJ}\lp\dfrac{\wt{X}^1\bs{\chi}_-^{k^-}\lp v_I^-,v_J^-\rp}{\nfi^+\cp_{\lambda}H}-\dfrac{\nabla_{I}^{\parallel}\wt{\psi}^-_{J}}{\wt{\nfi}^-}- \dfrac{\bs{\wt{\Theta}}^{\wt{L}^-}_-\lp v_{I}^-,v_{J}^-\rp}{\wt{\nfi}^-}\rp && \textup{in} && \tau^{11}, &\\
\label{bracket2}&  -\dfrac{\wt{X}^B\bs{\chi}_-^{k^-}\lp v_J^-,v_B^-\rp}{\nfi^+\cp_{\lambda}H}+\bs{\wt{\sigma}}^-_{\wt{L}^-}\lp v^-_J\rp && \textup{in} && \tau^{1I}, &
\end{align}

Since $\cp_{y^B}h^-_{IJ}=e^-_{B}\lp h^-_{IJ}\rp=\nabla^-_{v_B^-}\lp h^-_{IJ}\rp=2{\chr^-}_{B\lp J\rd}^Ah^-_{\ld I\rp A}-\frac{2}{\nfi^-}\psi_{\lp I\rd}^-\bs{\chi}^{k^-}_-\lp v^-_{\ld J\rp},v^-_B\rp$ (c.f. \eqref{eqA48}), the Christoffel symbols ${\Gamma^{\parallel}}^{A}_{BI}$ of the Levi-Civita covariant derivative $\nabla^{\parallel}$ of the metric $h^-_{IJ}$ are
\begin{align}
\label{ares2} h^-_{JA}{\Gamma^{\parallel}}^{A}_{BI}=&\textup{ }\dfrac{1}{2}\lp \cp_{y^B}h^-_{IJ}+\cp_{y^I}h^-_{BJ}-\cp_{y^J}h^-_{BI}\rp={\chr^-}_{BI}^Ah^-_{J A}-\frac{1}{\nfi^-}\psi_{J}^-\bs{\chi}^{k^-}_-\lp v^-_{I},v^-_B\rp.
\end{align}
Thus,
\begin{align}
\nonumber \bs{\wt{\Theta}}^{\wt{L}^-}_-\lp v_{I}^-,v_{J}^-\rp=&\textup{ }\la \nabla^-_{v^-_I}\wt{L}^-,v_{J}^-\ragm\\
\nonumber=&\textup{ }\la v^-_I\lp\alpha\rp L^-+\alpha\nabla^-_{v^-_I}L^-+\beta\nabla^-_{v^-_I}k^-+v_I^-\lp q^B\rp v_B^-+q^B\nabla^-_{v^-_I}v_B^-,v_{J}^-\ragm\\
\nonumber =&\textup{ }-\psi^-_J\cp_{y^I}\alpha +\alpha\bs{\Theta}^{L^-}_-\lp v_{I}^-,v_{J}^-\rp+\beta \bs{\chi}^{k^-}_-\lp v^-_I,v^-_J\rp  +h_{BJ}^-\cp_{y^I}q^B \\
\nonumber &\textup{ }+q^B\lp {\chr^-}_{BI}^Ah^-_{AJ}-\frac{1}{\nfi^-}\psi^-_J\bs{\chi}_-^{k^-}\lp v^-_I,v^-_B\rp\rp,\\
\nonumber =&\textup{ }-\psi^-_J\cp_{y^I}\alpha +\alpha\bs{\Theta}^{L^-}_-\lp v_{I}^-,v_{J}^-\rp+\beta \bs{\chi}^{k^-}_-\lp v^-_I,v^-_J\rp  +\nabla_{I}^{\parallel}q_J,
\end{align}
where in the third line we used Lemma \ref{lem1}. Now it follows directly from \eqref{nfipsiapB} that $\nabla_{I}^{\parallel}\wt{\psi}^-_{J}=\psi_J^-\cp_{y^I}\alpha+\alpha\nabla_{I}^{\parallel}\psi^-_{J}-\nabla_{I}^{\parallel}q_{J}$. Thus, we conclude
\begin{align}
\label{ares4} \bs{\wt{\Theta}}^{\wt{L}^-}_-\lp v_{I}^-,v_{J}^-\rp=\alpha\lp \bs{\Theta}^{L^-}_-\lp v_{I}^-,v_{J}^-\rp +\nabla_{I}^{\parallel}\psi^-_{J}\rp+\beta \bs{\chi}^{k^-}_-\lp v^-_I,v^-_J\rp  -\nabla_{I}^{\parallel}\wt{\psi}^-_{J},
\end{align}
and the invariance of \eqref{bracket1} follows from this expression and the second in \eqref{ares1}.
Concerning $\bs{\wt{\sigma}}^-_{\wt{L}^-}\lp v^-_J\rp$, we easily find
\begin{align}
\nonumber \bs{\wt{\sigma}}^-_{\wt{L}^-}\lp v^-_J\rp=&\textup{ }\dfrac{1}{\wt{\nfi}^-}\la \nabla^-_{v_J^-}k^-,\wt{L}^-\ragm=\dfrac{1}{\alpha\nfi^-}\la \nabla^-_{v_J^-}k^-,\alpha L^-+q^Bv^-_B\ragm\\
\label{ares5} =&\textup{ }\bs{\sigma}^-_{L^-}\lp v^-_J\rp+\dfrac{q^B}{\alpha\nfi^-}\bs{\chi}^{k^-}_-\lp v_J^-,v_B^-\rp.
\end{align}
Given that $\nfi^+\cp_{\lambda}H=\nfi^-A$, one obtains the invariance of \eqref{bracket2} by means of the first expression in \eqref{ares1}. This finishes the
consistency check of $\wt{\tau}^{ab}=\frac{1}{\alpha}\tau^{ab}$.

\section*{Acknowledgements}
The authors acknowledge financial support under the projects PGC2018-096038-B-I00 (Spanish Ministerio de Ciencia, Innovaci´on y Universidades and FEDER) and  SA096P20 (JCyL). M. Manzano also acknowledges the Ph.D. grant FPU17/03791 (Spanish Ministerio de Ciencia, Innovaci{\'o}n y Universidades).

\begingroup
\let\itshape\upshape

\bibliographystyle{acm}

\bibliography{ref}

\end{document}